\def\showcomments{0}
\documentclass[a4paper,USenglish,cleveref, numberwithinsect, autoref,pdfa, thm-restate]{lipics-v2021}

\pdfoutput=1 

\usepackage[utf8]{inputenc}

\usepackage{latexsym,amsmath,amsfonts,amssymb,stmaryrd,mathtools}
\usepackage{amsthm}
\usepackage{algorithm, algorithmicx}
\usepackage[noend]{algpseudocode}
\usepackage{graphicx}
\usepackage{hyperref}
\usepackage{listings, fancyvrb, multicol}
\usepackage{multirow}
\usepackage{comment}
\usepackage{xspace}
\usepackage{pifont}
\usepackage{parcolumns}
\usepackage{graphicx}
\usepackage{enumerate}
\usepackage{pgfplots}
\usetikzlibrary{patterns}

\newtheorem{invariant}[theorem]{Invariant}

\newcommand{\xparagraph}[1]{\subparagraph*{#1}}
\def\t{\textit}

\newcommand\rmv[1]{}
\def\broadcast{\t{broadcast}}
\def\deliver{\t{deliver}}

\def\scan{\t{scan}}

\def\prepare{\textsc{Prepare}\xspace}
\def\commit{\textsc{Commit}\xspace}
\def\vc{\textsc{ViewChange}\xspace}
\def\vcack{\textsc{ViewChangeAck}\xspace}
\newcommand{\neb}{Consistent Broadcast\xspace}
\newcommand{\rb}{Reliable Broadcast\xspace}
\newcommand{\combo}{Consistent and Reliable Broadcast\xspace}

\newcommand{\disk}{memory}

\definecolor{plum}{HTML}{75507b}	

\if\showcomments1

    \newcommand{\Naama}[1]{\noindent{\color{purple}{\textbf{Naama: }#1}}}
    \newcommand{\igor}[1]{\noindent{\color{brown}{\textbf{Igor: }#1}}}
    \newcommand{\mka}[1]{\noindent{\color{blue}{\textbf{MKA: }#1}}}
    \newcommand{\dalia}[1]{\noindent{\color{lipicsYellow}{\textbf{D: }#1}}}
    \newcommand{\ax}[1]{\noindent{\color{teal}{\textbf{Thn: }#1}}}
\else

     \newcommand{\Naama}[1]{}
     \newcommand{\igor}[1]{}
     \newcommand{\mka}[1]{}
     \newcommand{\dalia}[1]{}
     \newcommand{\ax}[1]{}
\fi

\makeatletter
\lst@Key{countblanklines}{true}[t]%
{\lstKV@SetIf{#1}\lst@ifcountblanklines}

\lst@AddToHook{OnEmptyLine}{%
	\lst@ifnumberblanklines\else%
	\lst@ifcountblanklines\else%
	\advance\c@lstnumber-\@ne\relax%
	\fi%
	\fi}
\makeatother

\title{Frugal Byzantine Computing\footnote{This paper is an extended version of the DISC 2021 paper.}}

\author{Marcos K. Aguilera}{VMware Research, Palo Alto, USA}{maguilera@vmware.com}{}{}
\author{Naama Ben-David}{VMware Research, Palo Alto, USA}{bendavidn@vmware.com}{}{}
\author{Rachid Guerraoui}{EPFL, Lausanne, Switzerland}{rachid.guerraoui@epfl.ch}{}{}
\author{Dalia Papuc}{EPFL, Lausanne, Switzerland}{dalia.papuc@epfl.ch}{}{}
\author{Athanasios Xygkis}{EPFL, Lausanne, Switzerland}{athanasios.xygkis@epfl.ch}{}{}
\author{Igor Zablotchi}{MIT, Cambridge, USA}{igorz@mit.edu}{}{}

\authorrunning{M.\,K. Aguilera, N. Ben-David, R. Guerraoui, D. Papuc, A. Xygkis, and I. Zablotchi} 

\Copyright{Marcos K. Aguilera, Naama Ben-David, Rachid Guerraoui, Dalia Papuc, Athanasios Xygkis, and Igor Zablotchi} 



\nolinenumbers 

\hideLIPIcs  

\EventEditors{John Q. Open and Joan R. Access}
\EventNoEds{2}
\EventLongTitle{42nd Conference on Very Important Topics (CVIT 2016)}
\EventShortTitle{CVIT 2016}
\EventAcronym{CVIT}
\EventYear{2016}
\EventDate{December 24--27, 2016}
\EventLocation{Little Whinging, United Kingdom}
\EventLogo{}
\SeriesVolume{42}
\ArticleNo{23}
 
\begin{document}
\maketitle

\begin{abstract}
Traditional techniques for handling Byzantine failures are
  expensive: digital signatures are too costly,
  while using $3f{+}1$ replicas is uneconomical ($f$ denotes the maximum number of Byzantine processes).
We seek algorithms that reduce the
  number of replicas to $2f{+}1$ and minimize the number of signatures. While the first goal can be achieved 
  in the message-and-memory model, accomplishing the second goal simultaneously is challenging.
We first address this challenge for the problem of broadcasting messages reliably.
We consider two variants of this problem, \neb{} and \rb, typically considered very close. 
Perhaps surprisingly, we establish a separation between them in terms of signatures required.
In particular, we show that \neb{} requires at least 1 signature in some execution, while
   \rb{} requires $O(n)$ signatures in some execution.
We present matching upper bounds for both primitives 
within constant factors.
We then turn to the problem of consensus and argue that this separation matters for solving consensus with
  Byzantine failures: we present a practical consensus algorithm that uses \neb{} as its main communication primitive.
This algorithm works for $n=2f{+}1$ and avoids signatures
  in the common-case---properties that have not been simultaneously
  achieved previously.
Overall, our work approaches Byzantine computing in a frugal manner and motivates the use of \neb{}---rather than \rb---as
  a key primitive for reaching agreement.
%
\end{abstract}
\section{Introduction}

Byzantine fault-tolerant computing is  notoriously expensive.  To tolerate $f$ failures, we typically need  $n=3f+1$ replica processes. 
Moreover, the agreement protocols for synchronizing the replicas have a significant latency overhead. Part
of the overhead comes from network delays, but digital signatures---often used in Byzantine computing---are
even more costly than network delays.
For instance, signing a message can be 28 times slower than sending it over a low-latency Infiniband fabric
(Appendix~\ref{app:latency} shows the exact measurements).

In this work, we study whether Byzantine computing can be \emph{frugal}, meaning if it can use
  few processes and few signatures. 
By Byzantine computing, we mean the classical problems of broadcast and consensus.
By frugality, we first mean systems with $n=2f+1$ processes, where $f$ is the maximum number of Byzantine processes. Such  systems are clearly preferable to systems with $n=3f+1$, as they require 33--50\% less hardware. 
However, seminal impossibility results imply that in the standard message-passing model with $n=2f+1$ processes, neither consensus nor various forms of broadcast can be solved, even under partial synchrony or randomization~\cite{dwork1988consensus}.
To circumvent the above impossibility results, we consider a message-and-memory (M\&M) model, which allows processes to
both pass messages and share memory, capturing the latest hardware capabilities of enterprise
servers~\cite{aguilera2018passing,aguilera2019impact}. In this model, it is possible to solve consensus with $n=2f+1$ processes and partial synchrony~\cite{aguilera2019impact}.




Frugality for us also means the ability to achieve  \emph{low latency}, by minimizing the number of digital signatures used. 
Mitigating the cost of digital signatures is commonly done by replacing them with more computationally efficient schemes, such as message authentication codes (MACs). For instance, with $n=3f+1$, the classic PBFT replaces some of its signatures with MACs~\cite{pbft2002}, while Bracha's broadcast algorithm~\cite{bracha1987asynchronous} relies exclusively on MACs. As we show in the paper, when $n=2f+1$, the same techniques for reducing the number of signatures are no longer applicable.

The two goals---achieving high failure resilience while minimizing the number of signatures---prove challenging when combined.
Intuitively, this is because with $n=2f+1$ processes, two quorums may intersect only at a Byzantine process; this is not the case with $n=3f+1$. Thus, we cannot rely on quorum intersection alone to ensure correctness; we must instead restrict the behavior of Byzantine processes to prevent them from providing inconsistent information to different quorums. Signatures can restrict Byzantine processes from lying, but only if there are enough correct processes to exchange messages and cross-check information. The challenge is to make processes prove that they behave correctly, based on the information they received so far, while using as few signatures as possible.

We focus initially on the problem of broadcasting a message reliably---one of the simplest and most widely used primitives in distributed computing.
Here, a designated sender process $s$ would like to send a message to other processes, such that all correct processes deliver the same message.
The difficulty is that a Byzantine sender may try to fool correct processes to deliver different messages.
Both broadcast variants, \combo{}, ensure that (1) if the sender is correct, then all correct processes deliver its message, and (2) any two correct processes that deliver a message must deliver the \textit{same} message. 
\rb ensures an additional property: if any correct process delivers a message, then all correct processes deliver that message.

Perhaps surprisingly, in the M\&M model we show a large separation between the two broadcasts in terms of the number of signatures (by correct processes) they require.
We introduce a special form  
of indistinguishability argument for $n=2f+1$ processes that uses signatures and shared memory in an elaborate way.
With it, we prove lower bounds for deterministic algorithms. 
For \neb{}, we prove that any solution requires one correct process to sign in some execution, and provide an algorithm that matches this bound.
In contrast, for \rb, we show that any solution requires at least $n-f-2$ correct processes to sign in some execution.
We provide an algorithm for \rb{} based on our \neb{} algorithm which follows the well-known Init-Echo-Ready pattern~\cite{bracha1987asynchronous} and uses up to $n+1$ signatures, matching the lower bound within a factor of 2.

To lower the impact of signatures on the latency of our broadcast algorithms, we introduce the technique of background signatures.
Given the impossibility of completely eliminating signatures, we design our protocols such that signatures are not used in well-behaved executions, i.e., when processes are correct and participate within some timeout.
In other words, both broadcast algorithms generate signatures in the background and also incorporate a fast path where signatures are not used.






We next show how to use our \neb algorithm to improve consensus algorithms.
The algorithm is based on PBFT~\cite{castro1999practical}, and maintains \textit{views} in which one process is the \textit{primary}. Within a view, agreement can be reached by simply having the primary consistent-broadcast a value, and each replicator respond with a consistent broadcast. When changing views, a total of $O(n^2)$ calls to \neb{} may be issued. The construction within a view is similar to our \rb{} algorithm. Interestingly, replacing this part with the \rb{} abstraction does \textit{not} yield a correct algorithm; the stronger abstraction hides information that an implementation based on \neb{} can leverage. For the correctness of our algorithm, we rely on a technique called \textit{history validation} and on \textit{cross-validating} the view-change message. Our consensus algorithm has four features: (1) it works for $n=2f+1$ processes, (2) it issues no signatures on the fast path, (3) it issues $O(n^2)$ signatures on a view-change and (4) it issues $O(n)$ background signatures within a view. As far as we know, no other algorithm achieves all these features simultaneously.
This result provides a strong motivation for the use of \neb{}---rather than \rb---as a first-class primitive in the design of agreement algorithms.

To summarize, we 
quantify the impossibility of avoiding signatures by proving lower bounds on the number of signatures required to solve the two variants of the broadcast problem---\textit{Consistent} and \textit{\rb{}}---and provide algorithms that match our lower bounds. Also, we construct a practical consensus algorithm using the \neb{} primitive.
In this work, we consider the message-and-memory model \cite{aguilera2018passing,aguilera2019impact}, but our results also apply to the pure shared memory model: our algorithms do not require messages so they work under shared memory,
  while our lower bounds apply \emph{a fortiori} to shared memory.


\section{Related Work}\label{sec:related}





    

\xparagraph{Message-and-memory models.}
We adopt a message-and-memory (M\&M) model, which is a generalization of both message-passing and shared-memory.
M\&M is motivated by enterprise servers with the latest hardware capabilities---such as RDMA, RoCE, Gen-Z, and soon CXL---which allow machines 
to \emph{both} pass messages and share memory.
M\&M was introduced by Aguilera et al. in~\cite{aguilera2018passing}, and subsequently studied in several other works~\cite{aguilera2019impact,attiya2020optimal,hadzilacos2020optimal,raynal2019one}.  Most of these works did not study Byzantine fault tolerance, but focused on crash-tolerant constructions when memory is shared only by subsets of processes~\cite{aguilera2018passing,attiya2020optimal,hadzilacos2020optimal,raynal2019one}. In~\cite{aguilera2019impact}, Aguilera et al. consider crash- and Byzantine- fault tolerance, as well as bounds on communication rounds on the fast path for a variant of the M\&M model with dynamic access permissions and memory failures. However, they did not study any complexity bounds off the fast path, and in particular did not consider the number of signatures such algorithms require.


\xparagraph{Byzantine Fault Tolerance.}
Lamport, Shostak and Pease~\cite{lamport1982byzantine,pease1980reaching} show that Byzantine agreement can be solved in synchronous message-passing systems iff $n \geq 3f+1$. 
In asynchronous systems subject to failures, consensus cannot be solved~\cite{fischer1985impossibility}. However, this result is circumvented by making
additional assumptions for liveness, such as randomization~\cite{ben1983another, MostefaouiMR14} or partial synchrony~\cite{chandra1996unreliable,dwork1988consensus}.
Even with signatures, asynchronous Byzantine agreement can be solved in message-passing systems only if $n\geq 3f+1$~\cite{bracha1985asynchronous}. Dolev and Reischuk~\cite{DolevR85} prove a lower bound of $n(f+1)/4$ signatures for Byzantine agreement, assuming that every message carries at least the signature of its sender.

\xparagraph{Byzantine Broadcast.}
In the message-passing model, both \combo{} require $n\geq 3f+1$ processes, unless (1) the system is synchronous and (2) digital signatures are available~\cite{bracha1985asynchronous,Dolev82,SrikanthT87}.
\neb{} is sometimes called Crusader Agreement~\cite{Dolev82}. The \neb{} abstraction was used implicitly in early papers on Byzantine broadcast~\cite{BrachaT83, Toueg84}, but its name was coined later by Cachin et al. in~\cite{CKPS2001}. The name ``consistent broadcast'' may also refer to a similar primitive used in synchronous systems~\cite{Lynch96,SrikanthT87}. 
Our \rb{} algorithm shares Bracha's Init-Echo-Ready structure~\cite{bracha1987asynchronous} with other broadcast algorithms~\cite{bracha1985asynchronous,echoMulticast,SrikanthT87}, but is the first algorithm to use this structure in shared memory to achieve \rb{} with $n=2f+1$ processes.



\xparagraph{BFT with stronger communication primitives.}
Despite the known fault tolerance bounds for asynchronous 
Byzantine Failure Tolerance (BFT), Byzantine consensus can be solved in asynchronous systems with $2f+1$ processes if stronger communication mechanisms are assumed.
%
Some prior work solves Byzantine consensus with $2f+1$ processes using specialized trusted components that Byzantine processes cannot control~\cite{ChunMS08,ChunMSK07,CorreiaNV04,correia2010asynchronous,KapitzaBCDKMSS12,VeroneseCBLV13}.
These trusted components can be seen as providing a broadcast primitive for communication.
These works assume the existence of such primitives as black boxes, and do not study the cost of implementing them using weaker hardware guarantees, as we do in this paper.
We achieve the same Byzantine fault-tolerance by using the shared memory to prevent the adversary from partitioning correct processes: once a correct process writes to a register, the adversary cannot prevent another correct process from seeing the written value.

It has been shown that shared memory primitives can be useful in providing BFT if they have \emph{access control lists} or \emph{policies} that dictate the allowable access patterns in an execution~\cite{aguilera2019impact,alon2005tight,bessani2009sharing,bouzid2016necessary,malkhi2003objects}. Alon et al.~\cite{alon2005tight} provide tight bounds for the number of strong shared-memory objects needed to solve consensus with optimal resilience. They do not, however, study the number of signatures required.

%



\xparagraph{Early termination.}
The idea of having a \emph{fast path} that allows early termination in well-behaved executions is not a new one, and has appeared in work on both message-passing~\cite{aguilera2019impact,aguilera2020microsecond,aublin2015next,dobre2006one,keidar2001cost,kotla2007zyzzyva,lamport2006fast} and shared-memory~\cite{balmau2016fast,timnat2014practical} systems.
Most of these works measure the fast path in terms of the number of message delays (or network rounds trips) they require, but some also consider the number of signatures~\cite{aublin2015next}. In this paper, we show that a signature-free fast path does not prevent an algorithm from having an optimal number of overall signatures.

\section{Model and Preliminaries}\label{sec:model}


We consider an asynchronous message-and-memory model, which allows processes to use
  both message-passing and shared-memory~\cite{aguilera2018passing}.
The system has $n$ processes $\Pi= \{p_1,\ldots, p_{n}\}$
  and a shared \emph{\disk} $M$.
Throughout the paper, the term \disk{} refers to $M$, not to the local state of processes.
We sometimes augment the system with eventual synchrony
   (\S\ref{sec:model:consensus}).

\xparagraph{Communication.} The memory consists of single-writer multi-reader (SWMR) read/write atomic registers. Each process can read all registers, and has access to an unlimited supply of registers it can write. If a process $p$ can write to a register $r$, we say that $p$ \emph{owns} $r$. 
This model is a special case of access control lists (ACLs)~\cite{malkhi2003objects}, and of dynamically permissioned memory~\cite{aguilera2019impact}. 
Additionally, every pair of processes $p$ and $q$ can send messages to each other over links that satisfy the \emph{integrity} and \emph{no-loss} properties.   Integrity requires that a message $m$ from $p$ be received by $q$ at most once
  and only if $m$ was previously sent by $p$ to $q$.
No-loss requires that a message $m$ sent from $p$ to $q$ be eventually received by $q$.



\xparagraph{Signatures.}
Our algorithms assume digital
signatures: each process has access to primitives to \textit{sign} and \textit{verify} signatures. A process $p$ may sign a value $v$, producing $\sigma_{p,v}$; we drop the subscripts when it is clear from context. Given $v$ and $\sigma_{p,v}$, a process can verify whether $\sigma_{p,v}$ is a valid signature of $v$ by $p$.

\xparagraph{Failures.} 
Up to $f$ processes may fail by becoming Byzantine, where $n=2f+1$.
Such a process can deviate
  arbitrarily from the algorithm, but cannot write on a register that is not its own, and cannot forge the signature of a correct process. As usual, Byzantine processes can collude, e.g., by using side-channels to communicate.
The memory $M$ does not fail; such a reliable memory is implementable
  from a collection of fail-prone memories~\cite{aguilera2019impact}.
We assume that these individual memories may only fail by crashing.
  


\subsection{Broadcast}
We consider two broadcast variants: \neb{}~\cite{RachidBook,CKPS2001} and \rb~\cite{Bracha84,RachidBook}. In both variants, broadcast is defined in terms of two primitives: $\broadcast(m)$ and $\deliver(m)$. A designated \emph{sender} process $s$ is the only one that can invoke $\broadcast$. When $s$ invokes $\broadcast(m)$ we say that \emph{$s$ broadcasts $m$}. When a process $p$ invokes $\deliver(m)$, we say that \emph{$p$ delivers $m$}.

\begin{definition}\label{def:neb}
\emph{\neb{}} has the following properties:
\begin{description}
    \item [Validity] If a correct process $s$ broadcasts $m$, then every correct process eventually delivers $m$.
    \item [No duplication] Every correct process delivers at most one message.
    \item [Consistency] If $p$ and $p'$ are correct processes, $p$ delivers $m$, and $p'$ delivers $m'$, then $m{=}m'$. 
    \item [Integrity] If some correct process delivers $m$ and $s$ is correct, then $s$ previously broadcast $m$.
\end{description}
\end{definition}

\begin{definition}\label{def:rb}
\emph{\rb{}} has the following properties:
\begin{description}
    \item [Validity, No duplication, Consistency, Integrity] Same properties as in Definition~\ref{def:neb}.
    \item [Totality] If some correct process delivers $m$, then every correct process eventually delivers a message.
\end{description}
\end{definition}

We remark that both broadcast variants behave the same way when the sender is correct and broadcasts $m$.
However, when the sender is faulty \neb{} has no delivery guarantees for correct processes, i.e., some correct processes may deliver $m$, others may not.
In contrast, \rb{} forces every correct process to eventually deliver $m$ as soon as one correct process delivers $m$.

\subsection{Consensus}\label{sec:model:consensus}
\begin{definition}
\emph{Weak Byzantine agreement~\cite{Lam83}} has the following properties:
\begin{description}
	\item [Agreement] If correct processes $i$ and $j$ decide $val$ and $val'$, respectively, then $val = val'$.
	\item [Weak validity]  If all processes are correct and some process decides $val$, then $val$ is the input of some process.
    \item [Integrity] No correct process decides twice.
    \item [Termination] Eventually every correct process decides.
\end{description}
\end{definition}

Our consensus algorithm (\S\ref{sec:consensus}) satisfies agreement, validity, and integrity under asynchrony, but requires eventual synchrony for termination.
That is, we assume that for each execution there exists a \textit{Global Stabilization Time (GST)}, unknown to the processes, such that from GST onwards there is a known bound $\Delta$ on communication and processing delays.

\section{Lower Bounds on Broadcast Algorithms}\label{sec:lowerbounds}

We show lower bounds on the number of signatures required to solve \combo{} with $n=2f+1$ processes in our model. We focus on signatures by correct processes because Byzantine processes can behave arbitrarily (including signing in any execution). 

\subsection{High-Level Approach}\label{sec:approach}

Broadly, we use indistinguishability arguments that
  create executions $E_v$ and $E_w$ that deliver different messages
  $v$ and $w$; then we create a composite execution $E$
  where a correct process cannot distinguish $E$ from $E_v$, while
  another correct process cannot distinguish $E$ from $E_w$, so they deliver
  different values, a contradiction.
  Such arguments are common in message-passing system, where the adversary can prevent communication by
  delaying messages between correct processes.
  However, it is not obvious how to construct this argument in shared memory, as the adversary cannot
  prevent communication via the shared memory, especially when using single-writer registers that cannot be overwritten by the adversary.
  Specifically, if correct processes write their values and read all registers, then for any two correct processes,
  at least one sees the value written by the other~\cite{bendavid2021unidir}. 
  So, when creating execution $E$ in which, say $E_v$ occurs first, processes executing $E_w$ will 
  know that others executed $E_v$ beforehand. 


  We handle this complication in two ways, depending on whether the sender signs its broadcast message. If the sender does not sign, we argue that processes executing $E_w$ cannot tell whether $E_v$ was executed by correct or Byzantine processes, and must therefore still output their original value $w$. This is the approach in the lower bound proof for \neb (Lemma~\ref{lem:onesig}).
  
  However, once a signature is produced, processes can save it in their memory to prove to others that they observed a valid signature. Thus, if the sender signs its value, then processes executing $E_w$ cannot be easily fooled; if they see two different values signed by the sender, then the sender is provably faulty, and correct processes can choose a different output. 
  So, we need another way to get indistinguishable executions. 
  We rely on a \emph{correct bystander} process.
  We make a correct process $b$ in $E$ sleep until all other correct processes decide. Then $b$ wakes up and observes that $E$ is a composition of $E_v$ and $E_w$. While $b$ can recognize that $E_v$ or $E_w$ was executed by Byzantine processes, it cannot distinguish which one. So $b$ cannot reliably output the same value as other correct processes. 
  We use this construction for \rb{}, but we believe it applies to other agreement problems in which all correct processes must decide.
  
  The proof is still not immediate from here. In particular, since $f {<} n/2$, correct processes can wait until at least $f{+}1$ processes participate
  in each of $E_v$ and $E_w$. Of those, in our proof
  we assume at most $f{-}1$ processes sign values. Since we need a bystander later, only $2f$ processes can participate. Thus, the sets executing $E_v$ and $E_w$ overlap at two processes; one must be the sender, to force decisions in both executions. Let $p$ be the other process and $S_v$ and $S_w$ be the set that execute $E_v$ and $E_w$ respectively, without the sender and $p$. Thus, $|S_v| = |S_w| = f{-}1$. 
  
  The key complication is that if $p$ signs its values in one of these two executions, we cannot compose them into an execution $E$ in which the bystander $b$ cannot distinguish which value it should decide. To see this, assume without loss of generality that $p$ signs a value in execution $E_w$. To create $E$, we need the sender $s$ and the set $S_w$ to be Byzantine. The sender will produce signed versions of both $v$ and $w$ for the two sets to use, and $S_w$ will pretend to execute $E_w$ even though they observed that $E_v$ was executed first. Since $|S_w| + |\{s\}| = f$, all other processes must be correct. In particular, $p$ will be correct, and will not produce the signature that it produces in $E_w$. Thus, the bystander $b$ will know that $S_v$ were correct. 
  More generally, the problem is that, while we know that at most $f-1$ processes sign, we do not know \emph{which} processes sign. A clever algorithm 
  can choose signing processes to defeat the indistinguishability argument---in our case, this happens if $p$ is a process that signs.
  
  Due to this issue, we take a slightly different approach for the \rb{} lower bound, first using the bystander construction to show 
  that any \rb{} algorithm must produce \emph{a single non-sender} signature. To strengthen this to our bound, we construct an execution in which this signature needs to be repeatedly produced. To make this approach work, we show not just that \emph{there exists} an execution in which a non-sender signature is produced, but that \emph{for all} executions of a certain form, a non-sender signature is produced. This change in quantifiers requires care in the indistinguishability proof, and allows us to repeatedly apply the result to construct a single execution that produces many signatures.

\subsection{Proofs}


In all proofs in this section, we denote by $s$ the designated sender process in the broadcast protocols we consider.
We first show that \neb{} requires at least one signature. 

\begin{lemma}\label{lem:onesig}
Any algorithm for \neb{} in the M\&M model with $n=2f+1$ and $ f\geq 1$ has an execution in which at least one correct process signs.
\end{lemma}


\begin{proof}
By contradiction, assume there is some algorithm $A$ for \neb in the M\&M model with $n=2f+1$ and $f\geq1$ without any correct process signing. 
Partition processes in $\Pi$ into 3 subsets: $S_1$, $S_2$, and $\{p\}$, where $S_1$ contains the sender, $\vert S_1 \vert = f$, $\vert S_2 \vert = f$, and $p$ is a single process.
Let $v,w$ be two distinct messages.
Consider the following executions.

\textsc{Execution $E_\textsc{clean-v}$.} Processes in $S_1$ and $p$ are correct
   (including the sender $s$), while processes in $S_2$ are faulty and never take a step.
Initially, $s$ broadcasts $v$. 
Since $s$ is correct, processes in $S_1$ and $p$ eventually deliver $v$. By our assumption that correct processes never sign, processes in $S_1$ and $p$ do not sign in this execution; processes in $S_2$ do not sign either, because they do not take any steps.

\textsc{Execution $E_\textsc{dirty-w}$.} Processes in $S_1$ and $S_2$ are correct 
  but $p$ is Byzantine.
Initially, $p$ sends all messages and writes to shared memory as it did in
  $E_\textsc{clean-v}$ (it does so without following its algorithm; $p$ is able to do this since no process signed in $E_\textsc{clean-v}$).
Then, the correct sender $s$ broadcasts $w$ and processes in $S_1$ and $S_2$ 
  execute normally, while $p$ stops executing.
Then, by correctness of the algorithm, eventually all correct processes deliver $w$. 
By our assumption that correct processes never sign, processes in $S_1$ and $S_2$ do not sign in this execution; $p$ does not sign either, because it acts as it did in $E_\textsc{clean-v}$.

\textsc{Execution $E_\textsc{bad}$.} Processes in $S_1$ are Byzantine, 
while processes in $S_2$ and $p$ are correct.
Initially, processes in $S_2$ sleep, while processes in $S_1$ and $p$
  execute, where processes in $S_1$ send the same messages to $p$ and write the same 
  values to shared memory as in $E_\textsc{clean-v}$ (but they do not send 
  any messages to $S_2$), so that from $p$'s perspective the
  execution is indistinguishable from $E_\textsc{clean-v}$. $S_1$ are able to do this because no process signed in $E_\textsc{clean-v}$.
Therefore, $p$ eventually
  delivers $v$. 
Next, processes in $S_1$ write the initial values to
  their registers\footnote{Recall
      that registers are single-writer. By ``their registers'',
      we mean the registers to which the processes can write.}.
Now, process $p$ stops executing, while processes in $S_1$ and $S_2$ execute
  the same steps as in $E_\textsc{dirty-w}$---here, note that
  $S_2$ just follows algorithm $A$ while $S_1$ is Byzantine and pretends to
  be in an execution where $s$ broadcasts $w$ ($S_1$ is able to do this because no process signed in $E_\textsc{dirty-w}$).
Because this execution is indistinguishable from $E_\textsc{dirty-w}$ to 
  processes in $S_2$, they eventually deliver $w$.
At this point, correct process $p$ has delivered $v$ while processes in $S_2$
  (which are correct) have delivered $w$, which contradicts the consistency
  property of \neb.
\end{proof}

An algorithm for \rb works for \neb, so Lemma~\ref{lem:onesig} also applies to \rb.

\medskip

We now show a separation between \neb{} and \rb{}: 
any algorithm for \rb{} has an execution where at least $f{-}1$ correct processes sign.

The proof for the \rb{} lower bound has two parts. First, we show that intuitively there are many executions in which some process produces a signature:
if $E$ is an execution in which 
  (1) two processes never take steps,
  (2) the sender is correct, and 
  (3) processes fail only by crashing, 
  then some non-sender process signs. This is the heart of the proof, and relies on the indistinguishability arguments discussed in Section~\ref{sec:approach}. 
  Here, we focus only on algorithms 
  in which at most $f$ correct processes sign, otherwise the algorithm
  trivially satisfies our final theorem.


\rmv{
Intuitively, the proof shows that if processes do not produce enough signatures, then we can construct an execution in which two processes go to sleep initially, while a Byzantine sender broadcasts $w$ and convinces correct processes
  excluding the sleeping ones to deliver $w$; then Byzantine processes
  change their state and memories to an execution where the sender broadcast $v$;
  then, we wake up the two sleeping processes, and the Byzantine processes
  convince those two processes to deliver $v$---a contradiction. }

\rmv{
  if this is not the case, 
  then we can construct an execution in which processes that were asleep at the beginning of the execution find themselves in a situation where either (a) a set $S$ could be Byzantine and the rest of the processes, in set $T$, want to deliver $v$ or (b) $T$ could be Byzantine and processes in $S$ want to deliver $w\neq v$. These processes do not know what to do.
}

\begin{lemma}\label{lem:atleastonesig}
Let $A$ be an algorithm for \rb{} in the M\&M model with $n=2f+1$ and $f\geq 2$ processes,
such that in any execution at most $f$ correct processes sign.
In all executions of $A$ in which at least $2$ processes crash initially, 
processes fail only by crashing, and the sender is correct, at least one correct non-sender process signs.
\end{lemma}

\def\executionOne{$E_\textsc{clean-v}$\xspace}
\def\executionTwo{$E_\textsc{clean-w}$\xspace}
\def\executionThree{$E_\textsc{mixed-v}$\xspace}
\def\executionFour{$E_\textsc{bad}$\xspace}

\begin{proof}
By contradiction, assume some algorithm $A$ satisfies the conditions of the lemma, but there is some execution of $A$ where the sender $s$ is correct, processes fail only by crashing,
and at least $2$ processes crash initially,
but no correct non-sender process signs.
Let \executionOne be such an execution, $D$ be a set with two processes that crash initially in \executionOne\footnote{If more than two processes crashed initially, pick any two arbitrarily.}, $C = \Pi \setminus D$,
and $v$ be the message broadcast by $s$ in \executionOne{}.
Consider the following executions:

\textsc{Execution \executionTwo.} The sender $s$ broadcasts some message $w\not=v$, $D$ crashes initially, and $C$ is correct. Since $s$ is correct, eventually all correct processes deliver $w$. 
By assumption, at most $f$ processes sign.
Let $S \subset C$ contain all processes that sign, augmented with any other processes so that $|S|=f$.
Let $T=C\setminus S$.
Note that (1) $|T|=f-1$ and (2) if $s$ signed, then $s \in S$, otherwise $s \in T$.

\textsc{Execution \executionOne.} This execution was defined above (where $s$ broadcasts $v$). Since $s$ is correct, eventually all correct processes deliver $v$.
At least one process in $T$ is correct---call it $p_t$---since processes in $D$ are 
  faulty and there are at least $f+1$ correct processes. 
Note that $p_t$ delivers $v$.
We refer to $p_t$ in the next execution.

\textsc{Execution \executionThree.} Processes in $S$ are Byzantine and 
  the rest are correct.
Initially, the execution is identical to \executionOne, except that
  (1) processes in $D$ are just sleeping not crashed, and 
  (2) processes in $S$ do not send messages to processes in $D$
     (this is possible because processes in $S$ are Byzantine).
The execution continues as in \executionOne until $p_t$ 
  delivers $v$. 
Then, 
processes in $S$ misbehave (they are Byzantine) 
  and do three things:
(1) they change their states to what they were at the end of \executionTwo (this is possible because no process in $T$ signed in \executionTwo),
(2) they write to their registers in shared memory
   the same last values that they wrote in \executionTwo, and
(3) they send the same messages they did in \executionTwo.
Intuitively, processes in $S$ pretend that $s$ broadcast $w$.
Let $t$ be the time at this point; we refer to time $t$
  in the next execution.
Now, we pause processes in $S$ and let all other processes execute, including $D$ which had been sleeping.
Since $p_t$ delivered $v$ and processes in $D$ are correct, they eventually deliver $v$ as well.

\textsc{Execution \executionFour.} 
Processes in $T\cup \lbrace s \rbrace$ are Byzantine and 
  the rest are correct.
Initially, the execution is identical to \executionTwo, except that
  (1) processes in $D$ are sleeping not crashed, and
  (2) processes in $T\cup \lbrace s \rbrace$ 
do not send messages to processes in $D$.
Execution continues as in \executionTwo until processes
   in $S$ (which are correct) deliver $w$.
Then, processes in $T\cup \{s\}$ misbehave and do three things:
(1) they change their states to what they were in \executionThree
    at time $t$---this is possible 
    because in \executionOne (and therefore in all values and messages they had by time $t$ in \executionThree), 
    no non-sender process signed, and in particular, there 
    were no signatures by any process in $S \setminus \{s\}$;
(2) they write to the registers in shared memory
    the same values that they have in \executionThree at time $t$; and
(3) they send all messages they did in \executionThree up to time $t$.
Intuitively, processes in $T\cup \{s\}$ pretend that $s$
  broadcast $v$.
Now, processes in $D$ start executing.
In fact, execution continues as in \executionThree
  from time $t$ onward, where processes is $S$ are paused
  and all other processes execute (including $D$).
Because these processes cannot distinguish the execution
  from \executionThree, eventually they deliver $v$.
Note that processes in $D$ are correct and they deliver $v$,
  while processes in $S$ are also correct and deliver $w$---contradiction.
 \end{proof}

In the final stage of the proof, we leverage Lemma~\ref{lem:atleastonesig} to construct an execution in which many processes sign. This is done by allowing some process to be poised to sign, and then pausing it and letting a new process start executing. Thus, we apply Lemma~\ref{lem:atleastonesig} $f-1$ times to incrementally build an execution in which $f-1$ correct processes sign.

\begin{theorem}\label{lem:atleastfsigs}
Every algorithm that solves \rb{} in the M\&M model with $n = 2f+1$ and $f\geq 1$ has some execution in which at least $f-1$ correct non-sender processes sign.
\end{theorem}
\begin{proof}
If $f=1$, the result is trivial; it requires $f-1=0$ processes to sign.

Now consider the case $f\geq 2$. If $A$ has an execution in which at least $f+1$ correct processes sign, then we are done. 
Now suppose $A$ has no execution in which at least $f+1$ correct processes sign.
Consider the following execution of $A$.

All processes and $s$ are correct. Initially, $s$ broadcasts $v$.
Then processes $s, p_1 \ldots p_{f}$ participate, and the rest are delayed. This execution is indistinguishable to $s, p_1 \ldots p_{f}$ from one in which the rest of the processes crashed. Therefore, by Lemma~\ref{lem:atleastonesig}, some process in $p_1 \ldots p_{f}$ eventually signs.
Call $p_1$ the first process that signs. We continue the execution until $p_1$'s next step is to make its signature visible. Then, we pause $p_1$, and let $p_{f+1}$ begin executing. Again, this execution is indistinguishable to processes $s, p_2 \ldots p_{f+1}$ from one in which the rest of the processes crashed, so by Lemma~\ref{lem:atleastonesig}, eventually some process in $p_2 \ldots p_{f+1}$ creates a signature and makes it visible. We let the first process to do so reach the state in which it is about to make its signature visible, and then pause it, and let $p_{f+2}$ start executing.


We continue in this way, each time pausing $p_i$ as it is about to make its signature visible, and letting $p_{f+i}$ begin executing. We can apply Lemma~\ref{lem:atleastonesig} as long as two processes have not participated yet. 
At that point, $f-1$ processes are poised to make their signatures visible. We then let these $f-1$ processes each take one step. This yields an execution of $A$ in which $f-1$ correct non-sender processes sign.
\end{proof}

\rmv{
The final theorem is an immediate corollary of Lemma~\ref{lem:atleastfsigs}, using the fact that we can extend the execution of Lemma~\ref{lem:atleastfsigs} with $f$ Byzantine signatures.

\begin{theorem}
Every algorithm that solves \rb{} on RDMA with $n = 2f+1$ and $f\geq 1$ has an execution in which at least $2f-1$ processes sign.
\end{theorem}

\begin{proof}
If $f=1$, the result is trivial because a Byzantine process can produce the required signature ($2f-1=1$).
If $f\geq 2$, the theorem follows from Lemma~\ref{lem:atleastfsigs}: for every algorithm $A$ that solves \rb{} on RDMA with $n = 2f+1$ processes, either $A$ has an execution in which $n$ processes sign, or, by Lemma~\ref{lem:atleastfsigs}, it has an execution in which at least $f-1$ correct processes sign. The former case establishes the result. In the latter case, Byzantine processes produce $f$ additional signatures for a total of $2f-1$.
\end{proof}
}

\section{Broadcast Algorithms}\label{sec:bcast}
In this section we present solutions for \combo{}.
We first implement \neb{} in Section~\ref{sec:neb}; then we use it as a building block to implement \rb, in Section~\ref{sec:rb}.
We prove the correctness of our algorithms in Appendix~\ref{app:cb-correctness} and~\ref{app:rb-correctness}. For both algorithms, we first describe the general execution outside the common case, which captures behavior in the worst executions; we then describe how delivery happens fast in the common case (without signatures).

\xparagraph{Process roles in broadcast.}
We distinguish between three process roles in our algorithms: sender, receiver, and replicator.
This is similar in spirit to the proposer-acceptor-learner model used by Paxos~\cite{lamport1998part}, and any process may play any number of roles.
If all processes play all three roles, then this becomes the standard model. 
The sender calls $\broadcast$, the receivers call $\deliver$, and the replicators help guarantee the properties of broadcast.
By separating replicators (often servers) from senders and receivers (often clients or other servers), we improve the practicality of the algorithms: clients, by not fulfilling the replicator role, need not remain connected and active to disseminate information from other clients.
Unless otherwise specified, $n$ and $f$ refer only to replicators; independently, the sender and any number of receivers can also be Byzantine.
Receivers cannot send or write any values, as opposed to the sender and replicators, but they can read the shared memory and receive messages.

\xparagraph{Background signatures.}
Our broadcast algorithms produce signatures in the background. We do so to allow the algorithms to be signature-free in the common-case. Indeed, in the common-case, receivers can deliver a message without waiting for background signatures. However, outside the common case, these signatures must still be produced by the broadcast algorithms in case some replicators are faulty or delayed. Both algorithms require a number of signatures that matches the bounds in Section~\ref{sec:lowerbounds} within constant factors.



\subsection{\neb}\label{sec:neb}


We give an algorithm for \neb{} that issues no signatures in the common case, when there is synchrony and no replicator is faulty. Outside this case, only the sender signs.

Algorithm~\ref{alg:fast-non-eq} shows the pseudocode. The broadcast and deliver events are called \textit{cb-broadcast} and \textit{cb-deliver}, to distinguish them from \textit{rb-broadcast} and \textit{rb-deliver} of \rb{}.
Processes communicate by sharing an array of \textit{slots}: process $i$ 
can write to \textit{slots[$i$]}, and can read from all slots.
To refer to its own slot, a processes uses index \textit{me}.
The sender $s$ uses its slot to broadcast its message while replicators use their slot to replicate the message. Every slot has two sub-slots---each a SWMR atomic register---one for a message (\textit{msg}) and one for a signature (\textit{sgn}).

To broadcast a message $m$, the sender $s$ writes $m$ to its \textit{msg} sub-slot (line~\ref{line:xswritem}). Then, in the background, $s$ computes its signature for $m$ and writes it to its \textit{sgn} sub-slot (line~\ref{line:xswrites}). The presence of \textit{msg} and \textit{sgn} sub-slots allow the sender to perform the signature computation in the background. Sender $s$ can return from the broadcast while this background task executes.

The role of a correct replicator is to copy the sender's message $m$ and signature $\sigma$, provided $\sigma$ is valid. The copying of $m$ and $\sigma$ (lines~\ref{line:xwhileinit}--\ref{line:xwhilefin}) are independent events, since a signature may be appended in the background, i.e., later than the message. The fast way to perform a delivery does not require the presence of signatures.
Note that correct replicators can have mismatching values only when $s$ is Byzantine and overwrites its memory.


A receiver $p$ scans the slots of the replicators. It delivers message $m$ when the content of a majority ($n{-}f$) of replicator slots contains $m$ and a valid signature by $s$ for $m$, and no slot contains a different message $m', m'\neq m$ with a valid sender signature (line~\ref{line:slow-return}). Slots with sender signatures for $m' \neq m$ result in a no-delivery. This scenario indicates that the sender is Byzantine and is trying to equivocate. Slots with signatures not created by $s$ are ignored so that a Byzantine replicator does not obstruct $p$ from delivering.
\begin{lstlisting}[columns=fullflexible,breaklines=true,keywords={while,if,else,return,do,for}, aboveskip=0pt, belowskip=0pt, float=t!,caption={\neb{} Algorithm with sender s},label={alg:fast-non-eq}]
Shared:
slots - @$n $@ array of "slots"; each slot is a 2-tuple (msg, sgn) of SWMR atomic registers, initialized to @$(\bot,\bot)$@.  

Sender code:
cb-broadcast(m):
    slots[me].msg.write(m) @\label{line:xswritem}@
    In the background:
        @$\sigma$@ = compute signature @for@ m
        slots[me].sgn.write(@$\sigma$@)@\label{line:xswrites}@

Replicator code:
while True: @\label{line:xwhileinit}@
    m = slots[s].msg.read()
    if m @$\neq$@ @$\bot$@:
        slots[me].msg.write(m)
    sign = slots[s].sgn.read()
    val = slots[me].msg.read()
    if val @$\neq$@ @$\bot$@ and sign @$\neq$@ @$\bot$@ and sign is a valid signature @for@ val:
        slots[me].sgn.write(sign) @\label{line:xwhilefin}@

Receiver code:
while True:
    others = scan() @\label{line:scan}@
    if others[i].msg has the same value m for all i in @$\Pi$@: // Fast path 
        cb-deliver(m); break @\label{line:fast-return}@
    if  others contains at least @$n-f$@ signed copies of the same value m
        and (@$\nexists $@i: others[i].sgn is a valid signature @for@ others[i].msg and others[i].msg @$\neq$@ m):  @\label{line:slow-check}@
        cb-deliver(m); break @\label{line:slow-return}@ 

scan(): @\label{line:scan-start}@
    others = [slots[i].(msg, sgn).read() for i in @$\Pi$@] @\label{line:fillOthers}@ 
    done = False
    while not done:
        done = True
        for i in @$\Pi$@:
            if others[i] == @$\bot$@:
                others[i] = slots[i].(msg, sgn).read() @\label{line:updateOthers}@
                if others[i] @$\neq$@ @$\bot$@:
                    done = False @\label{line:doneF}@
    return others @\label{line:scan-end}@
\end{lstlisting}

When there is synchrony and both the sender and replicators follow the protocol, a receiver delivers without using signatures.
Specifically, delivery in the fast path occurs when there is unanimity, i.e., all $n = 2f+1$ replicators replicated value $m$ (line~\ref{line:fast-return}), regardless of whether a signature is provided by $s$. A correct sender eventually appends $\sigma$, and $n-f$ correct replicators eventually copy $\sigma$ over, allowing another receiver to deliver $m$ via the slow path, even if a replicator misbehaves, e.g., removes or changes its value. \Naama{I think we should mention the fast path before the slow path. Then explain how even if someone delivered on the fast path and later Byz processes changed their values, others will never deliver something different}




An important detail is the use of a snapshot to read replicators' slots (line~\ref{line:scan}), as opposed to a simple collect. The scan operation is necessary to ensure that concurrent reads of the replicators' slots do not return views that can cause correct receivers to deliver different messages. To see why, imagine that the scan at line~\ref{line:scan} is replaced by a simple collect. Then, an execution is possible in which correct receiver $p_1$ reads some (correctly signed) message $m_1$ from $n-f$ slots and finds the remaining slots empty, while another correct receiver $p_2$ reads $m_2 \neq m_1$ from $n-f$ slots and finds the remaining slots empty. In this execution, $p_1$ would go on to deliver $m_1$ and $p_2$ would go on to deliver $m_2$, thus breaking the consistency property. We present such an execution in detail in Appendix~\ref{app:cb-correctness}. \dalia{//this is at the end of Appendix B; if someone can double check it, would be grateful; added it as 1 reviewer kindly asked for it} 

To prevent scenarios where correct receivers see different values at a majority of replicator slots, the \textit{scan} operation works as follows (lines~\ref{line:scan-start}--\ref{line:scan-end}): first, it performs a collect of the slots. 
If all the slots are non-empty, then we are done. 
Otherwise, we re-collect the \emph{empty slots} until no slot becomes non-empty between two consecutive collects. 
This suffices to avoid the problematic scenario above and to guarantee liveness despite $f$ Byzantine processes.

\subsection{\rb}\label{sec:rb}

We now give an algorithm for \rb{} that issues no signatures in the common case, and issues only $n+1$ signatures in the worst case. 
Algorithm~\ref{alg:rb-from-cb} shows the pseudocode.
%
%

Processes communicate by sharing arrays \textit{Echo} and \textit{Ready}, which
  have the same structure of sub-slots as \textit{slots} in Section~\ref{sec:neb}. 
\textit{Echo[$i$]} and \textit{Ready[$i$]} are writable only by replicator $i$, while the sender $s$ communicates with the replicators using an instance of \neb{} (CB) and does not access \textit{Echo} or \textit{Ready}. 
In this CB instance, $s$ invokes \textit{cb-broadcast}, acting as sender for CB, and the replicators invoke  \textit{cb-deliver}, acting as receivers for CB.

To broadcast a message, $s$ \textit{cb-broadcast}s $\langle$\textsc{Init},$m\rangle$ (line~\ref{line:ybdcst}). 
Upon delivering the sender's message $\langle$\textsc{Init},$m\rangle$, each replicator writes $m$ to its \textit{Echo} \textit{msg} sub-slot (line~\ref{line:ywecho}). 
Then, in the background, a replicator computes its signature for $m$ and writes it to its \textit{Echo} \textit{sgn} sub-slot (line~\ref{line:ywsigma}).
By the consistency property of \neb{}, if two correct replicators $r$ and $r'$ deliver $\langle$\textsc{Init},$m\rangle$ and $\langle$\textsc{Init}, $m'\rangle$ respectively, from $s$, then $m=m'$. Essentially, correct replicators have the same value or $\bot$ in their \textit{Echo} \textit{msg} sub-slot.

Next, replicators populate their \textit{Ready} slots with a \textit{ReadySet}.
A replicator $r$ constructs such a \textit{ReadySet} from the $n-f$ signed copies of $m$ read from the \textit{Echo} slots (lines~\ref{line:ybeginwait}--\ref{line:ReadySet}). 
In the background, $r$ reads the \textit{Ready} slots of other replicators and copies over---if $r$ has not written one already---any valid \textit{ReadySet} (line~\ref{line:ReadySetBkgrd}).
Thus, totality is ensured (Definition~\ref{def:rb}), as the \textit{ReadySet} created by any correct replicator is visible to all correct receivers. 
\begin{lstlisting}[columns=fullflexible,breaklines=true,keywords={if,when,return,else,for,break},aboveskip=0pt, belowskip=0pt,float=ht!,caption={Reliable Broadcast Algorithm with sender $s$},label={alg:rb-from-cb}]
Shared:
Echo, Ready - @$n $@ array of "slots"; each slot is a 2-tuple (msg, sgn) of SWMR atomic registers, initialized to @$(\bot,\bot)$@.  

Sender code:
rb-broadcast(m):
    cb-broadcast(<@\sc{Init}@,m>) @\label{line:ybdcst}@

Replicator code:
state = WaitForSender // @$\in$@{WaitForSender,WaitForEchos}
while True:
    if state == WaitForSender:
        if cb-delivered <@\sc{Init}@,m> from s:
            Echo[me].msg.write(m) @\label{line:ywecho}@
            In the background:
                @$\sigma = $@ compute signature @for@ m
                Echo[me].sgn.write(@$\sigma$@) @\label{line:ywsigma}@
            state = WaitForEchos
         
    if state == WaitForEchos:@\label{line:ybeginwait}@
        ReadySet = @$\emptyset$@
        for i @$\in \Pi$@:
            other = Echo[i].(msg,sgn).read()
            if other.msg == m and other.sgn is m validly signed by i:
                ReadySet.add((i,other))
    
        if size(ReadySet) @$\geq n-f$@:@\label{line:ReadySetSize}@
            ready = True
            Ready[me].msg.write(ReadySet)@\label{line:ReadySet}@

In the background:
    while True
        if not ready:
            others = [Ready[i].msg.read() for i in @$\Pi$@]
            if @$\exists$@i: others[i] is a valid ReadySet:
                ready = True
                Ready[me].msg.write(others[i])@\label{line:ReadySetBkgrd}@

Receiver code:
while True:
    others = [Echo[i].msg.read() for i in @$\Pi$@] 
    proofs = [Ready[i].msg.read() for i in @$\Pi$@] 
    if others contains @$n$@ matching values m: // Fast path
        rb-deliver(m); break @\label{line:rb-fast-return}@ 
    if proofs contains @$n-f$@ valid ReadySet @for@ the same value m:
        rb-deliver(m); break @\label{line:rb-slow-return}@ 
\end{lstlisting}
    %



To deliver $m$, a receiver $p$ reads $n-f$ valid \textit{ReadySet}s for $m$ (line~\ref{line:rb-slow-return}).\footnote{In contrast to Algorithm~\ref{alg:fast-non-eq}, receivers need not use the \textit{scan} operation when gathering information from the replicators' \textit{Ready} slots because there can only be a single value with a valid \textit{ReadySet} (Invariant~\ref{invariant:rbValidSets}).} This is necessary to allow a future receiver $p'$ deliver a message as well.
Suppose that $p$ delivers $m$ by reading a single valid \textit{ReadySet} $R$.\footnote{A similar argument that breaks totality applies if $p$ were to deliver $m$ by reading $n-f$ signed values of $m$ in the replicators' \textit{Echo} slots.}
Then, the following scenario prevents $p'$ from delivering: let sender $s$ be Byzantine and let $R$ be written by a Byzantine replicator $r$.
Moreover, let a \textit{single} correct replicator have \textit{cb-deliver}ed $m$, while the remaining correct replicators do not deliver at all, which is allowed by the properties of \neb{}.
So, the \textit{ReadySet} contains values from a single correct replicator and $f$ other Byzantine replicators.
If $r$ removes $R$ from its \textit{Ready} slot, it will block the delivery for $p'$ since no valid \textit{ReadySet} exists in memory.



A receiver $p$ can also deliver the sender's message $m$ using a fast path. 
The signature-less fast path occurs when $p$ reads $m$ from the \textit{Echo} slots of all replicators (line~\ref{line:rb-fast-return}), and the delivery of the \textsc{Init} message by the replicators is done via the fast path of \neb. 
This is the common-case, when replicators are not faulty and replicate messages timely.
Note that $p$ delivering $m$ via the fast path does not prevent another receiver $p'$ from delivering.
Process $p'$ delivers $m$ via the fast path if all the \textit{Echo} slots are in the same state as for $p$. Otherwise, e.g., some Byzantine replicators overwrite their \textit{Echo} slots, $p'$ delivers $m$ by relying on the $n-f$ correct replicators following the protocol (line~\ref{line:rb-slow-return}).

\section{Consensus}\label{sec:consensus}
We now give an algorithm for consensus using \neb{} as its communication primitive, rather than the commonly used primitive, \rb.
Our algorithm is based on the PBFT algorithm~\cite{castro1999practical,pbft2002} and proceeds in a sequence of (consecutive) views.  
It has four features: (1) it works for
$n=2f+1$ processes, (2) it issues no signatures in the common-case, (3) it issues $O(n^2)$ signatures on a view-change and (4) it issues $O(n)$ required background signatures within a view.


Our algorithm uses a sequence of \neb{} instances indexed by a broadcast sequence number $k$. When process $p$ broadcasts its $k^{\text{th}}$ message $m$, we say that $p$ broadcasts $(k, m)$. We assume the following ordering across instances, which can be trivially
guaranteed: ({\bf FIFO delivery})
 For $k \ge 1$, no correct process delivers $(k,m_k)$ from $p$ unless it has delivered $(i,m_i)$ from $p$, for all $i<k$.

Algorithm~\ref{alg:consensus} shows the pseudocode. Appendix~\ref{app:cons-correctness} has its full correctness proof.
The protocol proceeds in a sequence of consecutive \textit{views}.
Each view has a primary process, defined as the view number $\mathrm{mod}~n$ (line~\ref{line:primary}).
A view has two phases, \prepare and \commit. 
There is also
a view-change procedure initiated by a \vc message.

When a process is the primary (line~\ref{line:coord}), it broadcasts a \prepare message with its estimate \textit{init} (line~\ref{line:bdcastprepare}), which is either its input value or a value acquired in the previous view (line~\ref{line:n_inputVal}).
Upon receiving a valid \prepare message, a replica broadcasts a \commit message (line~\ref{line:ph2}) with the estimate it received in the \prepare message. We define a \prepare to be valid when it originates from the primary and either (a) \textit{view $= 0$} (any estimate works), or (b) \textit{view > $0$} and the estimate in the \prepare message has a proof from the previous view. Appendix~\ref{app:cons-validMsgs} details the conditions for a message to be valid.
When a replica receives an invalid \prepare message from the primary or times out,
it broadcasts a \commit message with $\bot$.
If a replica accepts a \prepare message with \textit{val} as estimate and $n-f$ matching \commit messages (line~\ref{line:n_decideCondition}), it decides on \textit{val}.
\begin{lstlisting}[columns=fullflexible,breaklines=true, keywords={loop, if, end, then, function, when, while,else, wait, until},aboveskip=0pt, belowskip=0pt,float=htb!,caption={Consensus protocol based on \neb{} ($n = 2f+1$)}, label={alg:consensus}]
propose(v@\subi@):
    view@\subi@ = @$0$@; est@\subi@ = @$\bot$@;  aux@\subi@ = @$\bot$@ 
    proof@\subi@ = @$\emptyset$@; vc@\subi@= @$(0,\bot,\emptyset)$@
    decided@\subi@ = False    
    while True:
        p@\subi@ = view@\subi@ @$\%\; n$@@\label{line:primary}@
        
        // Phase 1
        if p@\subi@ == i: @\label{line:coord}@
            init@\subi@ = est@\subi@ if est@\subi@ @$\neq \bot$@ else v@\subi@ @\label{line:n_inputVal}@
            cb-broadcast(<@\sc{Prepare}@, view@\subi@, init@\subi@, proof@\subi@>) @\label{line:bdcastprepare}@
        wait until receive valid <@\sc{Prepare}@, view@\subi@, val, proof> from p@\subi@ or timeout on p@\subi@ @\label{line:n_acceptPhase1}@
        if received valid <@\sc{Prepare}@, view@\subi@, val, proof> from p@\subi@:
            aux@\subi@ = val @\label{line:auxupdate}@
            vc@\subi@ = (view@\subi@,val,proof) @\label{line:vctuple}@
        else:
            aux@\subi@ = @$\bot$@ @\label{line:auxbot}@ 
        
        // Phase 2
        cb-broadcast(<@\sc{Commit}@, view@\subi@, aux@\subi@>)@\label{line:ph2}@
        wait until receive valid <@\sc{Commit}@, view@\subi@, *> from @$n-f$@ processes 
                   and (@$\forall$@j: receive valid <@\sc{Commit}@, view@\subi@, *> from j or timeout on j)@\label{line:waituntil_phase2}@
        @$\forall$@j: R@\subi@[j] = val if received valid <@\sc{Commit}@, view@\subi@, val> from j else @$\bot$@ @\label{line:set_Ri}@
        if @$\exists$@val @$\ne \bot:\#_{\text{val}}$@(R@\subi@)@$\geq n-f$@ and aux@\subi@ == val: @\label{line:n_decideCondition}@
            try_decide(val)@\label{line:n_decision}@
        
        // Phase 3
        cb-broadcast(<@\sc{ViewChange}@, view@\subi@ + 1, vc@\subi@>@$_{\sigma_i}$@)@\label{line:vc}@
        wait until receive @$n-f$@ non-conflicting view-change certificates for view@\subi@ + 1 @\label{line:acceptVC}@
        proof@\subi@ =  set of non-conflicting view-change certificates @\label{line:proof}@
        est@\subi@ = val in proof@\subi@ associated with the highest view @\label{line:n_EstVC}@
        view@\subi@ = view@\subi@ + 1 @\label{line:n_viewChange}@
    
    In the background:
        when cb-deliver valid <@\sc{ViewChange}@, view', vc>@$_{\sigma_j}$@ from j:
            cb-broadcast(<@\sc{ViewChangeAck}@, d>@$_{\sigma_i}$@) @\label{line:sendAck}@             // d is the view-change message being ACKed
        
try_decide(val):
    if not decided@\subi@:
        decided@\subi@ = True
        decide(val)@\label{line:n_finalDecide}@
\end{lstlisting}


The view-change procedure ensures that all correct replicas eventually reach a view with a correct primary and decide. 
It uses an acknowledgement phase similar
to PBFT with MACs~\cite{pbft2002}. While in~\cite{pbft2002} the mechanism is used so that the primary can prove the authenticity of a view-change message sent by a faulty replica, we use this scheme to ensure that (a) a faulty participant cannot lie about a committed value in its \vc{} message and (b) valid \vc{} messages can be received by all correct replicas.

A replica starts a view-change by broadcasting a signed \vc message with its view-change tuple (line~\ref{line:vc}). The view-change tuple \textit{(view, val, proof\textsubscript{val})} is updated when a replica receives a valid \prepare message (line~\ref{line:vctuple}). It represents the last non-empty value a replica accepted as a valid estimate and the view when this occurred. We use the value's proof, \textit{proof\textsubscript{val}}, to prevent a Byzantine replica from lying about its value: suppose a correct replica decides \textit{val} in view $v$, but in view $v+1$, the primary \textit{p} is silent, and so no correct replica hears from \textit{p}; without the proof, a Byzantine replica could claim to have accepted \textit{val$'$} in $v+1$ from \textit{p} during the view-change to $v+1$, thus overriding the decided value $\textit{val}$. 

When a replica receives a valid \vc message, it responds by broadcasting a signed \vcack containing the \vc message (line~\ref{line:sendAck}). A common practice is to send a digest of this message instead of the entire message~\cite{castro1999practical}.
We define a \vc message $m$ from $p$ to be valid when the estimate in the view-change tuple corresponds to the value broadcast by $p$ in its latest non-empty \commit and $m$'s proof is valid.
We point out that, as an optimization, this proof can be removed from the view-change tuple and be provided upon request when required to validate \vc messages. For instance, in the scenario described above, when a (correct) replica $r$ did not accept \textit{val}$'$ in view $v+1$, as claimed by the Byzantine replica $r'$, $r$ can request $r'$ to provide a proof for \textit{val}$'$. 

A view-change certificate consists of a \vc message and $n-f-1$ corresponding \vcack messages. This way, each view-change certificate has the contribution of at least one correct replica, who either produces the \vc message or validates a \vc message. Thus, when a correct replica $r$ receives a view-change certificate relayed by the primary, $r$ can trust the contents of the certificate.

To move to the next view, a replica must gather a set of $n-f$ non-conflicting view-change certificates $\Psi$. This step is performed by the primary of the next view, who then includes this set with its \prepare message for the new view. Two view-change certificates conflict if their view-change messages carry a tuple with different estimates ($\neq \bot$), valid proof, and same view number.  If the set $\Psi$ consists of tuples with estimates from different views, we select the estimate associated with the highest view. 
Whenever any correct replica decides on a value \textit{val} within a view, the protocol ensures a set of non-conflicting view-change certificates can be constructed only for \textit{val} and hence the value is carried over to the next view(s).





\subsection{Discussion}



We discuss how Algorithm~\ref{alg:consensus} achieves the four features mentioned at the beginning of Section~\ref{sec:consensus}. The first feature (the algorithm solves consensus with $n=2f+1$ processes) follows directly from the correctness of the algorithm. The second feature (the algorithm issues no signatures in the common-case) holds because in the common-case, processes will be able to deliver the required \prepare{} and \commit{} messages and decide in the first view, without having to wait for any signatures to be produced or verified. The third feature (the algorithm issues $O(n^2)$ signatures on view-change) holds because, in the worst case, during a view change each process will sign and broadcast a \vc{} message, thus incurring $O(n)$ signatures in total, and, for each such message, each other process will sign and broadcast a \vcack{} message, thus incurring $O(n^2)$ signatures. The fourth feature states that the algorithm issues $O(n)$ required background signatures within a view. These signatures are incurred by \textit{cb-broadcast}ing \prepare{} and \commit{} messages. In every view, correct processes broadcast a \commit{} message, thus incurring $n-f=O(n)$ signatures in total.

To the best of our knowledge, no existing algorithm has achieved all these four features simultaneously. The only broadcast-based algorithm which solves consensus with $n=2f+1$ processes that we are aware of, that of Correia et al.~\cite{correia2010asynchronous}, requires $O(n)$ calls to \rb{} before any process can decide; this would incur $O(n^2)$ required background signatures when using our \rb{} implementation---significantly more than our algorithm's $O(n)$ required background signatures.

At this point, the attentive reader might have noticed that our consensus algorithm uses some techniques that bear resemblance to our \rb{} algorithm in Section~\ref{sec:bcast}. Namely, the primary of a view \textit{cb-broadcast}s a \prepare{} message which is then echoed by the replicas in the form of \commit{} messages. Also, during view change, a replica's \vc{} message is echoed by other replicas in the form of \vcack{} messages. This is reminiscent of the Init-Echo technique used by our \rb{} algorithm. 

Thus, the following question arises: Can we replace each instance of the witnessing technique in our algorithm by a single \rb{} call and thus obtain a conceptually simpler algorithm, which also satisfies the three above-mentioned properties? Perhaps surprisingly, the resulting algorithm is incorrect. It allows an execution which breaks agreement in the following way: a correct replica $p_1$ \textit{rb-delivers} some value $v$ from the primary and decides $v$; sufficiently many other replicas time out waiting for the primary's value and change views without ``knowing about'' $v$; in the next view, the primary \textit{rb-broadcasts} $v'$, which is delivered and decided by some correct replica $p_2$.

Intuitively, by using a single \rb{} call instead of multiple \neb{} calls, some information is not visible to the consensus protocol. Specifically: while it is true that, in order for $p_1$ to deliver $v$ in the execution above, $n-f$ processes must echo $v$ (and thus they ``know about'' $v$), this knowledge is however encapsulated inside the \rb{} abstraction and not visible to the consensus protocol. Thus, the information cannot be carried over to the view-change, even by correct processes. This intuition provides a strong motivation to use \neb{}---rather than \rb---as a first-class primitive in the design of Byzantine-resilient agreement algorithms.

\section{Conclusion}
A common tool to address Byzantine failures is to use signatures or
  lots of replicas.
However, modern hardware makes these techniques prohibitive:
  signatures are much more costly than network communication, and
  excessive replicas are expensive.
Hence, we seek algorithms that minimize the number of signatures
  and replicas.
We applied this principle to broadcast primitives in a system that adopts the message-and-memory model,
  to derive algorithms that avoid signatures in the common case, use nearly-optimal
  number of signatures in the worst case, and require only $n=2f{+}1$ replicas.
We proved worst-case lower bounds on the number of signatures required by \neb{} and \rb{}, showing a separation between these problems.
We presented a Byzantine consensus algorithm based on our \neb{} primitive. 
This is the first consensus protocol for $n=2f{+}1$ without signatures 
  in the common case. 
A novelty of our protocol is the use of \neb instead of \rb, which resulted in 
  fewer signatures than existing consensus protocols based on \rb.

\bibliographystyle{plainurl}

\bibliography{references}

\appendix
\section{APPENDIX: Latency}\label{app:latency}

\begin{figure}[H]
    \centering
    \begin{tikzpicture}
      \tikzstyle{every node}=[font=\small]
      \begin{axis}[
        xbar,
        width=0.65\textwidth, height=4.5cm, 
        enlarge y limits=0.15,
        xmin=0, xmax=200, enlarge x limits={0.05},
        bar shift= 0pt,
        xlabel={Latency [$\mu s$]},
        symbolic y coords={Send a message using RDMA o/IB, Send a message using TCP o/IB,Sign a message using CPU, Send a message using TCP o/Ethernet,Sign a message using FPGA},
        ytick={Send a message using RDMA o/IB, Send a message using TCP o/IB,Sign a message using CPU, Send a message using TCP o/Ethernet,Sign a message using FPGA},
        nodes near coords, 
        nodes near coords align={horizontal},
        ]
        \addplot+[postaction={pattern=north east lines,pattern color=blue}] coordinates {(1.3,Send a message using RDMA o/IB) (14.120,Send a message using TCP o/IB) (53.442,Send a message using TCP o/Ethernet)};
        \addplot+[postaction={pattern=north west lines,pattern color=red}] coordinates {(36.1,Sign a message using CPU) (176.80,Sign a message using FPGA)};
      \end{axis}
    \end{tikzpicture}
    \caption{
    RDMA communication is significantly faster than signature creation using CPU
    or hardware acceleration (FPGA).
    The graph shows the latency of sending
    or signing a 32-byte message.
    IB means Infiniband, a faster interconnect
    than Ethernet found in data centers.
    TCP latencies are obtained using sockperf~\cite{sockperf}. RDMA latency is obtained using perftest~\cite{perftest}. Signatures use optimized implementations for CPU~\cite{secp256k1-impl} and FPGA~\cite{silex} of the ECDSA algorithm on the secp256k1 elliptic curve~\cite{secp256k1}. 
    An FPGA improves the throughput of signature creation (not
    shown in figure), but not its latency, due to their relatively low clock speeds (compared to CPUs) and the non-parallelizable nature of algorithms for digital signature.}
    \label{fig:sigs}
\end{figure}
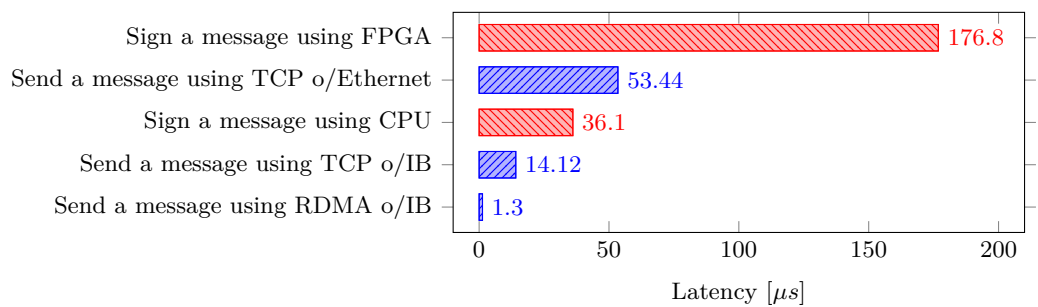

\section{APPENDIX: \neb{} Correctness}\label{app:cb-correctness}

We start with a simple observation:

\begin{observation}\label{obs:no-overwrite}
If $p$ is a correct process, then
no sub-slot that belongs to $p$ is written to more than once.
\end{observation}
\begin{proof}
Since $p$ is correct, $p$ never writes to any sub-slot more than once. 
Furthermore, since all sub-slots are single-writer registers, no other process can write to these sub-slots.
\end{proof}

Before proving that our implementation, Algorithm~\ref{alg:fast-non-eq}, satisfies the properties of \neb, we show two intermediary results with respect to the liveness and safety of the scan operation.

\begin{lemma}[Termination of scan.] Each $\scan$ operation returns.
\end{lemma}
\begin{proof}
We observe the following:
\begin{enumerate}
    \item If some slot $S$ goes from empty to non-empty between consecutive iterations of the while loop, then some process (the writer of slot $S$) wrote a value to $S$. This causes the while loop to continue (line~\ref{line:doneF}).
    \item Termination condition: the while loop exits (and thus the $\scan()$ operation returns) once either (a) \textit{others} has no empty slots, or (b) no slot goes from empty to non-empty between two consecutive iterations of the while loop (line~\ref{line:doneF} is never executed). 
    \item Once a slot $S$ is read and has non-empty value, \t{others} gets updated (lines~\ref{line:fillOthers} or~\ref{line:updateOthers}) and the slot $S$ is never read again.
    \item the size of $\t{others}$ is equal to the number of processes, $n$.
\end{enumerate}

By contradiction, assume the scan function never terminates.
This implies the scan function did not return after executing $n+1$ iterations of the while loop.
This means, each iteration at least one slot went from being empty to containing a value. This value is added to the $\t{others}$ array (line~\ref{line:updateOthers}) and the slot is not read again in future iterations.
Since the array $\t{others}$ is bounded by $n$, after $n$ iterations $\t{others}$ contains only non-empty values. 
At the $n+1^{th}$ execution of the loop, no slot is empty, $\t{done}$ stays true, and hence the operation returns. Contradiction.

Given every scan operation returns after executing at most $n+1$ times the while loop and each loop invokes at most $n$ reads of the base registers, the complexity of the $\scan()$ operation is within $O(n^2)$.
\end{proof}

\begin{lemma}[Non-inversion of scan]\label{lemma:inversionScan}
Let $p_1$ and $p_2$ be correct processes who invoke $\scan$. Let $V_1$ and $V_2$ be the return values of those scans, respectively. If $V_1$ contains some value $m_1$ in at least $n-f$ slots and $V_2$ contains some value $m_2$ in at least $n-f$ slots, then $V_1$ contains $m_2$ in at least one slot, or $V_2$ contains $m_1$ in at least one slot. 
\end{lemma}
\begin{proof}
Assume by contradiction that $V_1$ does not contain $m_2$ and $V_2$ does not contain $m_1$.

Since $V_1$ contains $m_1$ in at least $n-f$ slots, it must be that $m_1$ was written by at least one correct process; call this process $r_1$. Similarly, $m_2$ must have been written by at least one correct process; call this process $r_2$. Then the following must be true:
\begin{enumerate}
    \item $p_1$ must have read the slot of $r_2$ at least twice and found it empty. Let $t_{p_1\leftarrow r_2}$ be the linearization point of $p_1$'s last read of $r_2$'s slot before $p_1$ returns from the scan.
    \item $p_2$ must have read the slot of $r_1$ at least twice and found it empty. Let $t_{p_2\leftarrow r_1}$ be the linearization point of $p_2$'s last read of $r_1$'s slot before $p_2$ returns from the scan.
    \item $p_1$ must have read the slot of $r_1$ and found it to contain $m_1$. Let $t_{p_1\leftarrow r_1}$ be the linearization point of $p_1$'s last read of $r_1$'s slot.
    \item $p_2$ must have read the slot of $r_2$ and found it to contain $m_2$. Let $t_{p_2\leftarrow r_2}$ be the linearization point of $p_2$'s last read of $r_2$'s slot.
\end{enumerate}

We now reason about the ordering of $t_{p_1\leftarrow r_1}$, $t_{p_1\leftarrow r_2}$, $t_{p_2\leftarrow r_1}$, and $t_{p_2\leftarrow r_2}$:
\begin{enumerate}
    \item $t_{p_1\leftarrow r_1}$ < $t_{p_1\leftarrow r_2}$. Process $p_1$'s last read of $r_2$'s slot must have occurred during the last iteration of the while loop before returning from the scan. Furthermore, $p_1$'s last read of $r_1$'s slot cannot have occurred on the same last iteration of the loop, otherwise the non-empty read would have triggered another iteration; thus, $p_1$'s last read of $r_1$'s slot must have occurred either in a previous iteration at line~\ref{line:updateOthers}, or initially at line~\ref{line:fillOthers}.
    \item Similarly, $t_{p_2\leftarrow r_2}$ < $t_{p_2\leftarrow r_1}$. 
    \item $t_{p_1\leftarrow r_2}$ < $t_{p_2\leftarrow r_2}$. Process $p_1$'s last read of $r_2$ returns an empty value, while $p_2$'s last read of $r_2$ returns $m_2$. Since $r_2$ is correct, its slot cannot go from non-empty to empty, thus the empty read must precede the non-empty read.
    \item Similarly, $t_{p_2\leftarrow r_1}$ < $t_{p_1\leftarrow r_1}$.
\end{enumerate}

By transitivity, from (1)-(3), it must be that $t_{p_1\leftarrow r_1} < t_{p_2\leftarrow r_1} $. This contradicts (4). No valid linearization order exists for the four reads.
\end{proof}

We are now ready to prove that Algorithm~\ref{alg:fast-non-eq} satisfies the properties of \neb{}.

\begin{lemma}[Validity]
If a correct process $s$ broadcasts $m$, then every correct process eventually delivers $m$.
\end{lemma}

\begin{proof}
Let $s$ be a correct sender that broadcasts $m$ and consider a correct receiver $p$ that tries to deliver $s$'s message. 

Since $s$ is correct, it writes $m$ to its message sub-slot.
Therefore, all replicators read $m$ and no other message from $s$, by Observation~\ref{obs:no-overwrite}.

If all replicators are correct and they copy $m$ in a timely manner, then $p$ is able to deliver $m$ via the fast path (at line~\ref{line:fast-return}). 

Otherwise $s$, being correct, will eventually write its (valid) signature of $m$ to its signature sub-slot. Since we consider at most $f$ Byzantine processes that replicate $m$, the $n-f$ correct replicators are guaranteed to copy the signature of $m$ to their slot. Moreover, since $s$ is correct, and we assume Byzantine processes cannot forge the digital signatures of correct processes, no replicator can produce a different message $m'\neq m$ with $s$'s signature. This enables receiver $p$ to deliver $m$ via the slow path (at line~\ref{line:slow-return}).
\end{proof}

\begin{lemma}[No duplication]
Every correct process delivers at most one message. 
\end{lemma}
\begin{proof}
Correct processes only deliver at lines~\ref{line:fast-return} or \ref{line:slow-return}. Immediately after a correct $p$ process delivers a message, $p$ exits the \textit{while} loop and thus will not deliver again.
\end{proof}

\begin{lemma}[Consistency]
If $p$ and $p'$ are correct processes, $p$ delivers $m$ and $p'$ delivers $m'$, then $m{=}m'$. 
\end{lemma}

\begin{proof}

Assume by contradiction that consistency does not hold; assume correct process $p$ delivers $m$, while correct process $p'$ delivers $m' \ne m$. 

Assume first \textit{wlog} that $p$ delivers $m$ using the fast path. Then $p$ must have seen $m$ in $n$ replicator slots.
Assume now that $p'$ also delivers $m'$ using the fast path; then, $p'$ must have seen $m'$ in $n$ replicator slots. This means that all $n$ replicators must have changed their written value, either from $m$ to $m'$, or vice-versa; this is impossible since at least $n-f$ of the replicators are correct and never change their written value (Observation~\ref{obs:no-overwrite}).
Process $p'$ must have then delivered $m'$ using the slow path instead; then, $p'$ must have seen signed copies of $m'$ in $n-f$ replicator slots. This means that $n-f$ replicators, including at least one correct replicator, must have changed their value from $m$ to $m'$, or vice-versa; this is impossible by Observation~\ref{obs:no-overwrite}.

So it must be that $p$ and $p'$ deliver $m$ and $m'$, respectively, using the slow path. In this case, $p$ sees signed copies of $m$ in $n-f$ slots, while $p'$ sees signed copies of $m'$ in $n-f$ slots.  Lemma~\ref{lemma:inversionScan} therefore applies: $p$ must also see a signed copy of $m'$ or $p'$ must also see a signed copy of $m$. Given there exists another validly signed value, the check at line~\ref{line:slow-check} fails for $p$ or $p'$. We have reached a contradiction: $p$ does not deliver $m$ or $p'$ does not deliver $m'$.
\end{proof}




\begin{lemma}[Integrity] 
If some correct process delivers $m$ and $s$ is correct, then $s$ previously broadcast $m$.
\end{lemma}

\begin{proof}
Let a correct receiver $p$ deliver a value, say $m \neq \bot$. To deliver, $m$ must either be (a) the value $p$ reads from the slots of all replicators (line~\ref{line:fast-return}) or (b) the signed value $p$ reads from the slots of at least $n-f$ replicators (line~\ref{line:slow-return}).
In both cases, for the delivery of $m$ to occur, at least one correct replicator $r$ contributes by writing value $m$ (unsigned in case (a) or signed in case (b)) to its slot. Given $r$ is a correct process, it must have copied the value it read from the sender's slot. Furthermore, a correct sender never writes any value unless it \textit{cb-broadcasts} it. Therefore, $m$ must have been broadcast by the sender $s$.
\end{proof}

\subparagraph*{Execution.} We provide an example of an execution breaking consistency when the collect operation is used instead of the scan operation in Algorithm~\ref{alg:fast-non-eq}.
Let there be a Byzantine sender $s$ and $n=3$ replicators. Let $p_1, p_2$ be two correct receivers, $r_1, r_2$ two correct replicators and let replicator $r_3$ be Byzantine. Initially, let $s$ write $m_1$ signed in its slot. Let receiver $p_2$ start its collect. It reads the slot of $r_1$ which it finds empty, and sleeps. Let $r_1, r_3$ copy $m_1$ signed in their slot, while $r_2$ sleeps. Let $p_1$ perform its collect, find two signed copies of $m_1$ and deliver $m_1$ via the check at line~\ref{line:slow-return}. Let $s$ change its value to $m_2$ signed, while $r_3$, being Byzantine, changes its value to $m_2$ signed. We resume $r_2$ and let it copy $m_2$ signed. We resume $p_2$'s collect, continuing to read $r_2,r_3$ slots, seeing two values of $m_2$ signed (recall it previously read $r_1$'s slot while it was empty) and delivering $m_2$ via the check at line~\ref{line:slow-return}.
\section{APPENDIX: \rb{} Correctness}\label{app:rb-correctness}

\begin{invariant}\label{invariant:rbValidSets}
    Let $S$ and $S'$ be two valid ReadySets for $m$ and $m'$, respectively. Then, $m=m'$.
\end{invariant}
\begin{proof}
    By contradiction. Assume there exist valid \textit{ReadySets} $S$ and $S'$ for different values $m\ne m'$.
    Set $S$ (resp. $S'$) consists of at least $n-f$ signed $m$ (resp. signed $m'$) messages. 
    Then there exist correct replicators $r$ and $r'$ such that $r$ writes $m$ and its signature to its \textit{Echo} slot and $r'$ writes $m'$ and its signature to its \textit{Echo} slot. This is impossible since correct replicators only write $*$ in their \textit{Echo} slots once they have \textit{cb-delivered} $\langle$\textsc{Init},$*\rangle$ from the sender. By the consistency property of \neb, $m$ must be equal to $m'$.
\end{proof}

\begin{lemma}[Validity]
If a correct process $s$ broadcasts $m$, then every correct process eventually delivers $m$.
\end{lemma}

\begin{proof}
Assume the sender $s$ is correct and broadcasts $m$. Let $p$ be a correct receiver that tries to deliver $s$'s message. 

Since the sender is correct, it \textit{cb-broadcasts} $\langle$\textsc{Init},$*\rangle$.
By the validity property of \neb, all correct replicators will eventually deliver $\langle$\textsc{Init},$*\rangle$ from $s$. Then, all correct replicators will write $m$ to their \textit{Echo} message sub-slots, compute a signature for $m$ and write it to their \textit{Echo} signature sub-slots. 
If all replicators are correct and they copy $m$ in a timely manner, then $p$ is able to deliver $m$ via the fast path (at line~\ref{line:rb-fast-return}). 

All correct replicators will eventually read each other's signed messages $m$; thus every correct replicator will be able to either (a) create a valid \textit{ReadySet} and write it to its \textit{Ready} slot or (b) copy a valid \textit{ReadySet} to its \textit{Ready} slot. Thus, $p$ will eventually be able to read at least $n-f$ valid \textit{ReadySets} for $m$ and deliver $m$ via the slow path (at line~\ref{line:rb-slow-return}).    
\end{proof}

\begin{lemma}[No duplication]
Every correct process delivers at most one message. 
\end{lemma}
\begin{proof}
Correct processes only deliver at lines~\ref{line:rb-fast-return} or \ref{line:rb-slow-return}. Immediately after a correct $p$ process delivers a message, $p$ exits the \textit{while} loop and thus will not deliver again.
\end{proof}

\begin{lemma}[Consistency]
If $p$ and $p'$ are correct processes, $p$ delivers $m$ and $p'$ delivers $m'$, then $m{=}m'$. 
\end{lemma}
\begin{proof}
    By contradiction. Let $p,p'$ be two correct receivers. Let $p$ deliver $m$ and $p'$ deliver $m' \neq m$. We consider 3 cases: (1) $p$ and $p'$ deliver their messages via the fast path, (2) $p$ and $p'$ deliver their messages via the slow path, and (3) (\textit{wlog}) $p$ delivers via the fast path and $p'$ delivers via the slow path.
    \begin{enumerate}
        \item[(1)] $p$ and $p'$ must have delivered $m$ and $m'$ respectively, by reading $m$ (resp. $m'$) from the \textit{Echo} slots of $n$ replicators. Thus, there exists at least one replicator $r$ such that $p$ read $m$ from $r$'s \textit{Echo} slot and $p'$ read $m'$ from $r$'s \textit{Echo} slot. This is impossible since correct replicators never overwrite their \textit{Echo} slots.
        \item[(2)] $p$ and $p'$ must have each read $n-f$ valid \textit{ReadySets} for $m$ and $m'$, respectively. This is impossible by Invariant~\ref{invariant:rbValidSets}.
        \item[(3)] $p'$ read at least one valid \textit{ReadySet} for $m'$. To construct a valid \textit{ReadySet}, one requires a signed set of $n-f$ values for $m'$. Thus, at least one correct replicator $r$ must have written $m'$ to its \textit{Echo} slot and appended a valid signature for $m'$. Process $p$ delivered $m$ by reading $m$ from the \textit{Echo} slots of all $n$ replicators, which includes $r$. This is impossible since correct replicators never overwrite their \textit{Echo} slots.
    \end{enumerate}
\end{proof}

\begin{lemma}[Integrity] 
If some correct process delivers $m$ and $s$ is correct, then $s$ previously broadcast $m$.
\end{lemma}
\begin{proof}
    Let $p$ be a correct receiver that delivers $m$ and let the sender $s$ be correct. We consider 2 cases: (1) $p$ delivers $m$ via the fast path and (2) $p$ delivers $m$ via the slow path.
    \begin{enumerate}
        \item[(1)] Fast Path. $p$ must have read $m$ from the \textit{Echo} slot of at least one correct replicator $r$. Replicator $r$ writes $m$ to its slot only upon \textit{cb-delivering} $\langle$\textsc{Init},$m\rangle$ from $s$. By the integrity property of \neb, $s$ must have broadcast $m$.
        Moreover, a correct sender only invokes \textit{cb-broadcast}($\langle$\textsc{Init},$m\rangle$) upon a \textit{rb-broadcast} event for $m$.
        \item[(2)] Slow Path. $p$ must have read at least one valid \textit{ReadySet} for $m$. A \textit{ReadySet} consists of a signed set of $n-f$ values for $m$. Thus, at least one correct replicator $r$ must have written $m$ signed to its \textit{Echo} slot. The same argument as in case (1) applies.
    \end{enumerate}
\end{proof}

\begin{lemma}[Totality] 
If some correct process delivers $m$, then every correct process eventually delivers a message.
\end{lemma}    
\begin{proof}
    Let $p$ be a correct receiver that delivers $m$. We consider 2 cases: (1) $p$ delivers $m$ via the fast path and (2) $p$ delivers $m$ via the slow path.
    \begin{enumerate}
        \item[(1)] Fast Path. $p$ must have read $m$ from the \textit{Echo} slots of all $n$ replicators, which include $n-f$ correct replicators. These $n-f$ correct replicators must eventually append their signature for $m$. Every correct replicator looks for signed copies of $m$ in other replicators' \textit{Echo} slots. Upon reading $n-f$ such values, each correct replicator is able to construct and write a valid \textit{ReadySet} to its \textit{Ready} slot (or copy a valid \textit{ReadySet} to its \textit{Ready} slot from another replicator). Thus, every correct receiver will eventually read $n-f$ valid \textit{ReadySets} for $m$ and deliver $m$ via the slow path.
        \item[(2)] Slow Path. $p$ must have read valid \textit{ReadySets} for $m$ from the slots of $n-f$ replicators, which must include at least one correct replicator $r$. Since $r$ is correct, $r$ will never remove its \textit{ReadySet} for $m$. Thus, every correct replicator will eventually either (a) copy $r$'s \textit{ReadySet} to their own \textit{Ready} slots or (b) construct and write a \textit{ReadySet} to their \textit{Ready} slots. Note that by Invariant~\ref{invariant:rbValidSets}, all valid \textit{ReadySets} must be for the same value $m$. Thus, every correct receiver will eventually read $n-f$ valid \textit{ReadySets} for $m$ and deliver $m$ via the slow path.
    \end{enumerate}
\end{proof}    

\section{APPENDIX: Byzantine Consensus Correctness and Additional Details}\label{app:cons-correctness}
\subsection{Valid messages}\label{app:cons-validMsgs}
A $\langle$\prepare, \textit{view, val, proof}$\rangle$ message is considered valid by a (correct) process if:
\begin{itemize}
    \item the process is part of \textit{view},
    \item the broadcaster of the \prepare is the coordinator of \textit{view}, i.e., \textit{view}$\;\%\;n$,
    \item when \textit{view} $=0$, \textit{proof}$=\emptyset$ and \textit{val} can be any value $\neq \bot$
    \item when \textit{view} $>0$, the estimate matches the highest view tuple in \textit{proof} and the \textit{proof} set is valid, i.e., it contains a set of $n-f$ non-conflicting view-change certificates for view \textit{view}; in case all tuples in \textit{proof} are still the init value $(0,\bot,\emptyset)$, any estimate is a valid estimate,
    \item the process did not previously accept a different \prepare in \textit{view}.
\end{itemize}

A $\langle$\commit, \textit{view, val}$\rangle$ message is considered valid by a (correct) process if:
\begin{itemize}
    \item the process is part of \textit{view},
    \item \textit{val} can be any estimate,
    \item the broadcaster did not previously send a view change message for \textit{view$'$ $>$ view},
    \item the broadcaster did not previously send another \commit message for \textit{val$'$ $\neq$ val} in the same \textit{view}. 
\end{itemize}

A $\langle$\vc, \textit{view$+1$, (view\textsubscript{val}, val, proof\textsubscript{val})}$\rangle$ message from process $j$ is considered valid by a (correct) process if:
\begin{itemize}
    \item \textit{val} $\in$ \textit{(view\textsubscript{val}, val, proof\textsubscript{val})} corresponds to the latest non-empty value broadcast in a $\langle$\commit, \textit{view\textsubscript{c}, val\textsubscript{c}}$\rangle$, \textit{val} $=$ \textit{val\textsubscript{c}} and \textit{view\textsubscript{val}} $=$ \textit{view\textsubscript{c}} ($\leq$ \textit{view}) and \textit{proof\textsubscript{val}} is a valid proof for \textit{val} (either consists of non-conflicting certificates that support \textit{val} as highest view-tuple or all tuples are with their init value; all \vc and \vcack messages must be for v\textit{iew\textsubscript{val}}),
    \item \textit{val} $\in$ \textit{($0$, val, proof\textsubscript{val})}, \textit{proof\textsubscript{val}} is $\emptyset$,
    \item if for each view \textit{view}$'$ $\leq$ \textit{view}, $\langle$\commit, \textit{view}$'$, $\bot\rangle$ from $j$ are empty; then \textit{(view\textsubscript{val}, val, proof\textsubscript{val})} must be equal to $(0,\bot,\emptyset)$,
    \item $j$ must have sent a single \commit message each view \textit{view}$'$ $\leq$ \textit{view},
    \item $j$ did not send another \vc message this view, \textit{view}$+1$.
\end{itemize}

\subsection{Agreement}
\begin{lemma}\label{invariant:n_sameview}
In any view \textit{v}, no two correct processes accept \textup{\prepare} messages for different values \textit{val} $\neq$ \textit{val$'$}.
\end{lemma}

\begin{proof}
Let $i,j$ be two correct processes. Any correct process accepts a \prepare messages only from the current view's primary (line~\ref{line:n_acceptPhase1}). 

A correct primary $p$ never broadcasts conflicting \prepare messages (i.e., same view \textit{v}, but different estimates \textit{val, val$'$, val $\neq$ val$'$}). This means, $i,j$ must receive the same \prepare message. 
By Lemma~\ref{lemma:always_valid}, both $i$ and $j$ consider the \prepare message from $p$ valid.


A faulty primary $p'$ may broadcast conflicting \prepare messages. 
Assume the primary broadcasts ($k$, $\langle$\prepare, \textit{v, val, proof}$\rangle$) and ($k'$,$\langle$\prepare, \textit{v, val$'$, proof$'$}$\rangle$) where $k, k'$ are the broadcast sequence numbers used. We distinguish between the following cases:
\begin{enumerate}
    \item $k < k'$: By the FIFO property, any correct replica must process message $k$ of $p'$ before processing message $k'$. If process $i$ accepts the $k^{\text{th}}$ message of $p'$, following the consensus protocol, $i$ will not accept a second \prepare message in the same view \textit{v}, i.e. message $k'$. Similarly for correct replica $j$.\label{2.1:  n_case1}
    \item $k > k'$: The argument is similar to~(\ref{2.1:  n_case1}).
    \item $k = k'$: In this case, $p'$ equivocates. If a message gets delivered by both $i$ and $j$, then the message is guaranteed to be the same by the consistency property of \neb.
\end{enumerate}
We conclude correct replicas agree on the \prepare message accepted within the same view.
\end{proof}

\begin{lemma} \label{invariant:n_decideinview}
In any view \textit{v}, no two correct processes call \textit{try\_decide} with different values \textit{val} and \textit{val$'$}.
\end{lemma}

\begin{proof}

By contradiction. 
Let $i,j$ be two correct processes. Assume in view \textit{v}, processes $i,j$ call \textit{try\_decide} with value \textit{val}, respectively \textit{val$'$}.
To call \textit{try\_decide}, the condition at line~\ref{line:n_decideCondition} must be true for both $i$ and $j$. This means $i$, (resp. $j$) accepts a valid \prepare message supporting \textit{val} (resp. \textit{val$'$}) and a set of $n-f$ \commit messages supporting \textit{val} (resp. \textit{val$'$}). By Lemma~\ref{invariant:n_sameview}, correct processes cannot accept different \prepare messages and consequently cannot call \textit{try\_decide} with different values since \textit{aux$_i$ = aux$_j$}.
\end{proof}

\begin{lemma}  \label{invariant:n_sameEstVal}
Let a correct process $i$ decide \textit{val} in view \textit{v}. For view \textit{v+1}, no valid proof can be constructed for a different estimate \textit{val$'$ $\neq$ val}.
\end{lemma}
\begin{proof}
By contradiction.
Let $i$ decide \textit{val} in view \textit{v}. 
Assume the contrary and let there be a valid proof such that \textit{(v+1, val$'$, proof)}. 

Given \textit{v+1} $> 0$, \textit{proof} cannot be $\emptyset$. It must be the case that the proof supporting \textit{val$'$} consists of a set of $n-f$ non-conflicting view-change certificates. Each view-change certificate consists of a \vc message with format $\langle$\vc, \textit{v+1}, \textit{(view, value, proof\textsubscript{val})}$\rangle$ and $f$ corresponding \vcack messages.
Any view-change certificate requires the involvement of at least one correct replica, namely, either a correct replica is the broadcaster of a \vc message or a correct replica validates a \vc message, by sending a corresponding \vcack. 

For \textit{val$'$} to be consistent with \textit{proof}, \textit{proof} must contain either (a) at least one view-change certificate with tuple \textit{(v, val$'$, proof\textsubscript{val$'$})} and no other view-change certificate s.t. its tuple has a different value for the same view, \textit{v}, i.e., $\not\exists$ \textit{(v, val, proof\textsubscript{val})}, with \textit{v} the highest view among the $n-f$ tuples or (b) only view-change certificates with tuples having the initial value $(0,\bot,\emptyset)$ so that any value is a valid value. 
Let $R_1$ denote the set of processes that contributed with a \vc message, which is then part of a view-change certificate in \textit{proof}.

Given $i$ decided \textit{val} in view \textit{v}, $i$ received $n-f$ \commit messages for \textit{val} (line~\ref{line:n_decideCondition}). Such processes must have received a valid \prepare message and updated their view-change tuple together with their auxiliary in lines~\ref{line:auxupdate} and~\ref{line:vctuple}, before sending a \commit message. Let $R_2$ denote the set of processes that contributed with a \commit message for \textit{val}.

These two sets, $R_1$ and $R_2$, must intersect in at least one replica $j$.
Replica $j$ must have used \neb{} for its view-change message: ($k_{vc}$, $\langle$\vc, \textit{v+1, (v, val$'$, proof\textsubscript{val$'$})}$\rangle$); the argument is similar  for the case ($k_{vc}$, $\langle$\vc, \textit{v+1, $(0,\bot,\emptyset)$})$\rangle$, where $k_{vc}$ is the broadcast sequence number used; otherwise it could have not gathered enough \vcack{s}, since correct replicas do not accept messages not delivered via the broadcast primitive. Similarly, $j$ must have used \neb{} for its \commit message: ($k_{c}$, $\langle$\commit, \textit{v, val}$\rangle$), where $k_{c}$ is the broadcast sequence number used; otherwise $i$ would not have accepted the \commit message.

If $j$ is correct, and sends a \commit message for \textit{val}, it broadcasts a \vc message with its true estimate, \textit{val}. Hence, the $R_1$ set of non-conflicting view-change messages must contain a tuple \textit{(v, val, proof\textsubscript{val})}. This yields either a set of conflicting view-change certificates if $\exists$ another view-change certificate for \textit{(v, val$'$, proof\textsubscript{val$'$})}, or a conflict between \textit{proof} and \textit{val$'$} as matching estimate (since \textit{v} is the highest-view and the value associated with this tuple corresponds to estimate \textit{val} and not \textit{val$'$}).

If $j$ is Byzantine, we distinguish between the following cases:
\begin{enumerate}
    \item $k_{c}$ < $k_{vc}$ ($j$ broadcasts its \commit message before it broadcasts its \vc message). In this case, no correct process sends a \vcack for $j$'s \vc message. By the FIFO property, a correct process first delivers the $k_{c}$ message and then $k_{vc}$ message. In order to validate a \vc message, the last non-empty value broadcast in a \commit must correspond to the value broadcast in the \vc. Since these do not match, no correct process sends a \vcack for $j$'s \vc. Hence, the \vc message of $j$ does not gather sufficient ACKs to form a view-change certificate and be included in \textit{proof}.\\
    \textit{Note:} If $j$ were to broadcast two \commit messages in view \textit{v}, one supporting \textit{val} and another supporting \textit{val$'$} before broadcasting its \vc message supporting \textit{val$'$}, no correct process ACKs its \vc message since $j$ behaves in a Byzantine manner, i.e., no correct process broadcasts two (different) \commit messages within the same view.
    \item $k_{vc}$ < $k_{c}$ ($j$ broadcasts its \vc message before it broadcasts its \commit message). In this case, process $i$ must have first delivered the \vc message from $j$. Consequently, $i$ does not accept $j$'s \commit message as valid. This contradicts our assumption that $i$ used this \commit message to decide \textit{val}.
    \item $k_{vc}=k_{c}$ ($j$ equivocates). By the properties of \neb, correct processes either deliver $j$'s \commit message, case in which the \vc message does not get delivered by any correct replica, and consequently does not gather sufficient \vcack to form a view-change certificate (for neither \textit{val$'$} nor $\bot$); or correct processes deliver $j$'s \vc message, case in which the \commit message does not belong to $R_2$, $i$ does not decide.
\end{enumerate}

We conclude, if $i$ decided \textit{val} in view \textit{v}, no valid proof can be constructed for view \textit{v+1} and \textit{val$'$ $\neq$ val}.
\end{proof}




\begin{lemma}\label{invariant:n_sameval}
Let a correct process $i$ decide \textit{val} in view \textit{v}. For any subsequent view \textit{v$'$ $>$ v}, no valid proof can be constructed for a different estimate \textit{val$'$ $\neq$ val}.
\end{lemma}
\begin{proof}
We distinguish between the following two cases: (1) \textit{v$'$ $=$ v+1} and (2) \textit{v$'$ $>$ v}.

\textbf{Case 1:} Follows from Lemma~\ref{invariant:n_sameEstVal}.

\textbf{Case 2:} By contradiction.
Let process $i$ decide \textit{val} in view \textit{v}.
Assume the contrary and let \textit{v$'$ $>$ v} be the lowest view in which there exists a valid proof for \textit{val$'$ $\neq$ val}, i.e., \textit{(v$'$,val$'$, proof)}. 

A valid proof supporting \textit{val$'$} must contain $n-f$ non-conflicting view-change certificates out of which (a) one view-change certificate supports \textit{val$'$} or (b) all view-change certificates claim $\bot$.
A view-change certificate consists of a \vc message and $f$ \vcack messages. This means, at least one correct process must validate a \vc message by broadcasting a \vcack message, or be the producer of a \vc message.

(a) For \textit{val$'$} to be the representative value of the $n-f$ view-change certificates in \textit{proof}, one of the view-change messages must contain a tuple with the highest view among all $n-f$ tuples. Let this tuple be \textit{(v$'-1$, val$'$, proof\textsubscript{val$'$})} such that \textit{v$'-1$} is the highest view possible before entering \textit{v$'$}.
This tuple must then come from view \textit{v$'-1$} with a valid proof, \textit{proof\textsubscript{val$'$}} supporting the fact that \textit{val$'$} is a valid value.

By assumption, the only valid proof that can be constructed in views prior to \textit{v$'$} but succeeding \textit{v} is for estimate \textit{val}. Hence, there is no valid proof, \textit{proof\textsubscript{val$'$}} for \textit{val$'$} in view \textit{v$'-1$}. 
In the case in which the producer of the \vc is correct, it will not construct a \vc message with an invalid proof. The \textit{vc\textsubscript{i}} variable is only updated if the \prepare message is valid.
In the case in which the producer of the \vc is faulty, it will not gather the necessary \vcack to form a view-change certificate given correct replicas do not validate a \vc message with an invalid proof or in which the estimate value contradicts the proof.
This contradicts our initial assumption that there exists a valid view-change certificate supporting \textit{val$'$}.

(b) All \vc messages in \textit{proof} have tuples $(0,\bot,\emptyset)$ so that any estimate value is valid value. Let this set be denoted by $R_1$.
Since $i$ decided \textit{val} in view \textit{v}, a set of $R_2$ replicas contributed with a \commit value for \textit{val}.
Sets $R_1$ and $R_2$ must intersect in one replica, say $j$.  
If $j$ sends \commit messages in subsequent views for a value $\bot$, $j$ must send a \vc message matching its latest non-empty \commit message, i.e., $\langle$\vc, \textit{v$'$, (v, val, proof\textsubscript{val}})$\rangle$, in order to gather sufficient \vcack and hence form a view-change certificate.
If $j$ sends a \commit message in any subsequent view for a value \textit{val$'$} $\neq \bot$, the only possible valid proof is for value \textit{val}, see case (a). 
Whichever the case, at least one view-change certificate in \textit{proof} must contain a view-change message with a non-empty tuple which contradicts our assumption that all view-change certificates are for $\bot$.
\end{proof}

\begin{theorem}[Agreement]
If correct processes $i$ and $j$ decide \textit{val} and \textit{val$'$}, respectively, then \textit{val $=$ val$'$}.\end{theorem}

\begin{proof}
We distinguish two cases: (1) decision in the same view (2) decision in different views.

\textbf{Case 1:} decision in the same view. Follows from Lemma~\ref{invariant:n_decideinview}.

\textbf{Case 2:} decision in different views. 
By contradiction. Let $i,j$ be two correct processes. Assume processes $i$ and $j$ decide two different values, \textit{val}, respectively \textit{val$'$}, in views \textit{v}, respectively \textit{v$'$}. Let \textit{v $<$ v$'$} \textit{wlog}.

To decide, a correct process must receive a valid \prepare message and $n-f$ \commit messages for the same estimate, line~\ref{line:n_decideCondition}. 
When $i$ decides \textit{val} in view \textit{v}, by Lemma~\ref{invariant:n_sameval}, from view \textit{v+1} onward, the only valid \textit{proof} supports estimate \textit{val}. Hence, a valid \prepare message can only contain an estimate for \textit{val} and at any view-change procedure, no \vc supporting \textit{val$'$} is able to form a view-change certificate.
Given process $j$ only accepts valid \prepare messages, $j$ cannot adopt \textit{val$'$} as its auxiliary, \textit{aux$_j$}. This means $j$ cannot decide \textit{val$'$}. 
Given process $j$ only collects a set of (non-conflicting) view-change certificates, $j$ cannot adopt \textit{val$'$} as its estimate \textit{est$_j$}.
\end{proof}




\subsection{Integrity}
\begin{theorem}[Integrity]
No correct process decides twice.
\end{theorem}
\begin{proof}
A correct process may call \textit{try\_decide} (line~\ref{line:n_decision}) multiple times. Yet, once a correct process calls \textit{decide} (line~\ref{line:n_finalDecide}), the \textit{decided} variable is set to \textit{true} and hence the if statement is never entered again.
\end{proof}

\subsection{Validity}
\begin{theorem}[Weak validity]
If all processes are correct and some process decides \textit{val}, then \textit{val} is the input of some process.
\end{theorem}
\begin{proof}
Assume a correct process decides \textit{val}. Following the steps in the algorithm, a correct process only decides a value for which it receives a valid \prepare message and $n-f$ \commit messages, in the same view (line~\ref{line:n_decideCondition}). It is either the case the value in the \prepare message comes from the previous view or it is the input value of the current view's primary (line~\ref{line:n_inputVal}). For the latter, validity is satisfied. For the former, the value in the previous view must come from one of the \vc messages. Which is either an input value of a prior view's primary or the value of a previous view message. 
We continue by applying the same argument inductively, backward in the sequence of views, until we reach a view in which the value was the input value of a primary. This shows that \textit{val} was proposed by the primary in some view.
\end{proof}

\subsection{Termination}
\begin{lemma}\label{lemma:vc_non-conflicting}
Two correct processes cannot send conflicting \vc messages.
\end{lemma}
\begin{proof}
Assume the contrary and let \textit{v} be the earliest view in which correct processes $i$ and $j$ send conflicting \vc messages $m_1$ and $m_2$, respectively. Let \textit{vc$_i$ $=$ (view$_i$, val$_i$, proof$_i$)} and \textit{vc$_j$ $=$ (view$_j$, val$_j$, proof$_j$)} be the view-change tuples in $m_1$ and $m_2$, respectively. Since $m_1$ and $m_2$ conflict, it must be the case that \textit{view$_i$ $=$ view$_j$} and \textit{$\bot\ne $ val$_i\ne$ val$_j \ne \bot$}. 
Thus, in view \textit{view$_i$ $=$ view$_j$}, $i$ and $j$ must have received and accepted \prepare messages for different values \textit{val$_i$} and \textit{val$_j$}. This contradicts Lemma~\ref{invariant:n_sameview}.
\end{proof}

\begin{lemma}\label{lemma:always_valid}
A \prepare, \commit or \vc message from a correct process is considered valid by any correct process.
\end{lemma}
\begin{proof}

A correct process $i$ only sends a \prepare message if it is the coordinator of that view (line~\ref{line:coord}).
When \textit{view} $=0$, \textit{est}$_i$ is initialized to $\bot$ which leads $i$ to set \textit{init}$_i$ to \textit{v}$_i$ (line~\ref{line:n_inputVal}). The \prepare message has the following format: $\langle$\prepare, $0$, \textit{v}$_i$, $\emptyset\rangle$ which matches the required specification for a valid \prepare.
When \textit{view} $>0$, any correct process updates its \textit{proof}$_i$ and \textit{est}$_i$ before increasing its \textit{view}$_i$ variable, i.e. moving to the next view.
A correct process would update these two vars according to the protocol, lines~\ref{line:proof} and~\ref{line:n_EstVC}.
As before, in case \textit{est}$_i = \bot$, $i$ to set \textit{init}$_i$ to \textit{v}$_i$, otherwise it carries \textit{est}$_i$ (line~\ref{line:n_inputVal}). The \prepare message has the following format: $\langle$\prepare, \textit{view$_i$, init$_i$, proof$_i$}$\rangle$ which matches the required specification for a valid \prepare.


A correct process $i$ broadcasts exactly one \commit message in \textit{view} (line~\ref{line:ph2}) after it either (a) hears from the coordinator of the current view or (b) starts suspecting the coordinator.
In case (a) $i$'s message contains the estimate of the coordinator (line~\ref{line:auxupdate}), while in case (b) it contains $\bot$ (line~\ref{line:auxbot}).
In any of the two cases, $i$'s \vc message strictly follows the \commit message (lines~\ref{line:vc} and~\ref{line:ph2}).
The behaviour is in-line with the specification.

A correct process $i$ broadcasts exactly a single \vc message in one \textit{view} (line~\ref{line:vc}) with its \textit{vc}$_i$. Process $i$ update its view-change tuple, \textit{vc}$_i$, only when it receives a valid \prepare message. Such message is ensured to be in accordance with the prior specifications for a valid \prepare message. Notice that a valid \prepare message cannot be $\bot$, and hence \textit{vc}$_i$ is either its initial value, $(0,\bot,\emptyset)$ or a valid tuple \textit{(view, val, proof)}.
The data in \textit{vc}$_i$ is updated at the same time \textit{aux}$_i$ is updated, upon receiving a valid \prepare, and these two variables indicate the same estimate (lines~\ref{line:auxupdate} and~\ref{line:vctuple}). The \textit{aux}$_i$ is then send via a \commit message within the same view (line~\ref{line:ph2}).
This ensures that the broadcast of $i$'s latest non-empty \commit corresponds to the data in its \textit{vc}$_i$ variable.
\end{proof}


Let $i$ be a correct process. For a given execution $E$, we denote by $\mathcal{V}(i)$ the set of views in which $i$ enters. We denote $v_{max}(i) = \max \mathcal{V}(i)$; by convention $v_{max}(i)=\infty$ if $\mathcal{V}(i)$ is unbounded from above.

\begin{lemma}\label{lemma:infinite_views}
For every correct process $i$, $v_{max}(i)=\infty$
\end{lemma}
\begin{proof}
Assume the contrary and let \textit{wlog} $i$ be the process with the lowest $v_{max}$. Since $i$ never progresses past view $v_{max}(i)$, $i$ must be blocked forever in one of the \textit{wait until} statements at lines~\ref{line:n_acceptPhase1}, \ref{line:waituntil_phase2}, or \ref{line:acceptVC}. We now examine each such case:
\begin{enumerate}
    \item Line~\ref{line:n_acceptPhase1}. If the primary $p$ of view $v_{max}(i)$ is faulty and does not broadcast a valid \prepare{} message, then eventually $i$ times out on the primary and progresses past the \textit{wait until} statement. If $p$ is correct, then $p$ eventually reaches view $v_{max}(i)$ and broadcasts a \prepare message $m$. By the validity property of Consistent Broadcast, $i$ eventually delivers $m$ from $p$. By Lemma~\ref{lemma:always_valid}, $i$ considers $m$ valid and thus progresses past the \textit{wait until} statement.
    \item Line~\ref{line:waituntil_phase2}. By our choice of $i$, every correct process must eventually reach view $v_{max}(i)$. Given the argument at item (1) above, no correct process can remain blocked forever at the \textit{wait until} statement in line~\ref{line:n_acceptPhase1}, thus every correct process eventually broadcasts a \commit message in view $v_{max}(i)$. By the validity property of Consistent Broadcast and by Lemma~\ref{lemma:always_valid}, $i$ eventually delivers all such messages and considers them valid. Therefore, $i$ must eventually deliver valid \prepare messages from $n-f$ processes and progress past the \textit{wait until} statement.
    \item Line~\ref{line:acceptVC}. By our choice of $i$, every correct process must eventually reach view $v_{max}(i)$. Given the argument at items (1) and (2) above, no correct process can remain blocked forever at the \textit{wait until} statements in lines~\ref{line:n_acceptPhase1} and \ref{line:waituntil_phase2}, thus every correct process eventually broadcasts a \vc message in view $v_{max}(i)$. By the validity property of Consistent Broadcast and by Lemma~\ref{lemma:always_valid}, every correct process eventually delivers all such \vc messages and considers them valid. Thus, for every \vc message $m$ sent by a correct process in view $v_{max}(i)$, every correct process eventually broadcasts a \vcack message $mAck$ with $m$'s digest; furthermore, $i$ receives and considers valid each such $mAck$. Thus, $i$ eventually gathers a set of $n-f$ view-change certificates in view $v_{max}(i)$, which are non-conflicting by Lemma~\ref{lemma:vc_non-conflicting}. This means that $i$ is eventually able to progress past the \textit{wait until} statement.
\end{enumerate}
We have shown that $i$ cannot remain blocked forever in view $v_{max}(i)$ in any of the \textit{wait until} statements. Thus, $i$ must eventually reach line~\ref{line:n_viewChange} and increase \textit{view}$_i$ to $v_{max}(i)+1$. We have reached a contradiction. 
\end{proof}

We define a view \textit{v} to be \textit{stable} if in \textit{v}: (1) the coordinator is correct and (2) no correct process times out on another correct process.



\begin{theorem}[Termination]
Eventually every correct process decides.
\end{theorem}
\begin{proof}
We will show that every correct process eventually calls \textit{try\_decide}, which is sufficient to prove the result.
By our assumption of eventual synchrony, there is a time $T$ after which the system is synchronous. 
We can also assume that after $T$, no correct process times out on another process. 
Let $i$ be a correct process. 
Let $v^*$ be the earliest view such that: (1) $i$ enters $v^*$ after time $T$ and (2) the primary of $v^*$ is correct. 
Recall that by Lemma~\ref{lemma:infinite_views}, $i$ and all other correct processes are guaranteed to eventually reach view $v^*$. 
Let $p$ be the (correct) primary of $v^*$. By our choice of $v^*$, $p$ broadcasts a \prepare message $m$ in $v^*$, which is received and considered valid by all correct processes (by the validity property of Consistent Broadcast and Lemma~\ref{lemma:always_valid}). Thus all correct processes will set their \textit{aux} variable to the value \textit{val} contained in $m$, and broadcast a \commit message with \textit{val}. Process $i$ must eventually deliver these \commit messages and consider them valid, thus setting at least $n-f$ entries of $R_i$ to \textit{val} in line~\ref{line:set_Ri}. Therefore, the check at line~\ref{line:n_decideCondition} will succeed for $i$ and $i$ will call \textit{try\_decide} at line~\ref{line:n_decision}.
\end{proof}

\end{document}